\documentclass[11pt,reqno]{amsart}
\usepackage[utf8]{inputenc}
\usepackage[english,russian]{babel}
\usepackage{xcolor}

\usepackage{amsthm}
\usepackage{caption}
\usepackage{booktabs}
\usepackage{amsmath}
\usepackage{amssymb}
\usepackage{amsfonts}
\usepackage{graphicx}
\usepackage{float}
\usepackage{algorithm}
\usepackage{algpseudocode}
\usepackage{xcolor}
\usepackage[colorlinks]{hyperref}
\usepackage{tikz}
\usepackage{pgfplots}
\usepackage{enumitem}
\usepackage{cleveref}

\pgfplotsset{compat=1.18}
\usetikzlibrary{decorations.pathreplacing,calligraphy,patterns,arrows.meta}
\newlength{\dk} \dk=1.7em
\newlength{\dktmp} \dktmp=0.7\dk
\tikzset{x=1.0\dk,y=1.0\dk,
 cell/.style={circle,draw=black, thin, minimum size=0.77\dk},
 }

\usepackage{ulem}
\usepackage{xcolor}
\usepackage{titlesec}
\usepackage[a4paper,left=3cm,right=3cm,top=3cm,bottom=3cm]{geometry}

\titlespacing{\section}{0pt}{2.5ex}{1.5ex}
\titlespacing{\subsection}{0pt}{1.5ex}{1ex}
\titlespacing{\subsubsection}{0pt}{1.5ex}{1ex}
\titleformat{\section}{\large\bfseries\centering}{\thesection}{1em}{}
\titleformat{\subsection}[runin]{\bfseries}{\thesubsection.}{0.5em}{}[.\mbox{\ }]
\titleformat{\subsubsection}[runin]{\bfseries}{\thesubsubsection.}{0.4em}{}[.\mbox{\ }]

\definecolor{darkgreen}{RGB}{0,70,0}
\newcommand{\greensquare}{\textcolor{green}{\rule{1.7ex}{1.7ex}}}

\usepackage{xcolor}

\newcommand{\bluesquare}{%
  \begingroup
    \setlength{\unitlength}{1.7ex}%
    \color{blue}%
    \begin{picture}(1,1)
      \linethickness{0.2pt}%
      \put(0,0){\framebox(1,1){}}%
      \put(0.02,0.02){\color{white}\rule{0.96\unitlength}{0.96\unitlength}}%
    \end{picture}%
  \endgroup
}
\theoremstyle{definition}
\newtheorem{lemma}{Lemma}
\newtheorem{remark}{Remark}
\newtheorem{theorem}{Theorem}

\newtheorem{corollary}{Corollary}
\newtheorem{proposition}{Proposition}
\newtheorem{example}{Example}
\newtheorem{construction}{Construction}
\newtheorem{approach}{Approach}

\hypersetup{linkcolor=blue,filecolor=black,urlcolor=blue, citecolor=blue}

\begin{document}
\renewcommand{\refname}{References}
\renewcommand{\proofname}{Proof.}
\renewcommand{\figurename}{Fig.}

\thispagestyle{empty}

\title[Minimum distances of LDPC codes in 5G standard]{\large Minimum distances of LDPC codes in 5G standard}

\author[V.R. DANILKO]{{\bf V.R. Danilko}}
\author[I.YU. MOGILNYKH]{{\bf I.Yu. Mogilnykh}}

\author[YA.A. TIKHOMOLOV]{{\bf Ya.A. Tikhomolov}}

\thanks{\rm V. R. Danilko and Ya. A. Tikhomolov are with Novosibirsk State University, Novosibirsk, Russia and I. Yu. Mogilnykh is with the Sobolev Institute of Mathematics, Novosibirsk, Russia. The work  was performed according to the Government
research assignment for IM SB RAS, Project No. FWNF-2026-0011.}

\vspace{1cm}
\maketitle {\small

\vspace{-12pt}

\bigskip
\bigskip

\begin{quote}
\noindent\textbf{Abstract:}    
We propose several approaches for bounding the minim\-um distances of the family of 
quasi-cyclic LDPC codes in the 5G NR standard.  In particular, we show that the high-rate $[9984, 8448]$ and the low-rate 
$[25344, 8448]$ BG1 5G LDPC codes have minimum distances in the ranges $\{8,\ldots, 14\}$ and 
$\{22,\ldots, 57\}$, respectively. Also we  propose a new early termination approach 
based on circulant modular reduction, which significantly lowers syndrome calculation complexity for the 
LDPC decoder.
\end{quote}
}

\bigskip

\section{Introduction}

Quasi-cyclic LDPC codes used in the 5G NR standard \cite{5G} form a distinctive class of error-correcting codes that combine strong performance with efficient hardware implementation.  In this paper, we consider the minimum distance problem for 5G LDPC codes, which is known to be NP-hard for the general class of linear codes \cite{Vardy}. Special attention is given to the codes formed by the first $4$ and $6$ parity circulant blocks, referred to hereafter as the $4$-layer and $6$-layer codes, respectively, for the following reasons.

The full 5G LDPC parity-check matrix \cite{Richardson} has an unusual structure for a quasi-cyclic LDPC code: it is obtained from four parity layers by appending a sparse matrix below them and a large identity matrix in the lower-right corner, see Fig.~\ref{fig:bg1-bg2-mask}. On one hand, this structure makes encoding very  computationally efficient \cite{5Genc}. On the other hand, one may expect that the presence of weight one columns in the parity-check matrix degrades the minimum distance of the code and, consequently, its decoding performance. However, due to the careful choice of circulant blocks in the matrix construction \cite{Richardson} and the retransmission mechanism of the 5G communication system, this degradation is significantly reduced when a low-complexity approximate decoder such as the layered min-sum decoder is used, although it remains visible for high-rate codes; see Fig.~\ref{fig:ibler_intro}.Moreover, the decoder is prone to undetected errors, see Fig.~\ref{FigET_UIBLER}. This phenomenon can be explained by the short minimum distances of high-rate 5G LDPC codes, which we establish in this study.

The $4$-layer code also plays a central role in the 5G HARQ retransmission scheme. In this scheme, although only certain parts of a codeword are transmitted in each retransmission, the initial transmission (redundancy version 0, RV0) starts at the third circulant block and includes all the remaining information bits and partially the leftmost parity bits. This fact makes RV0 closely related to the $4$-layer code.

The $[9984,8448]$ code, which is the $6$-layer code with maximum circulant size $384$, was also adopted as part of the SDA-OCT communication standard \cite{SDA}. In this standard, the error-detection properties are strengthened. Each information block contains a $144$-bit header followed by a $16$-bit CRC, and the block is additionally protected by $800$ parity bits of a convolutional code \cite[Section 3.4.5.2]{SDA}. As a result, the information-block error-detection capabilities of this standard are much stronger than those in 5G, partly to cope with the use of a simpler retransmission ARQ protocol in SDA-OCT.

The next section presents the necessary preliminaries, including basic definitions, notations, and a description of 5G LDPC codes.

The main goal of this work is to derive upper and lower bounds on the minimum distances of the 5G LDPC codes. To this end, in Section~\ref{Sec_Vont} we exploit the block structure of 5G LDPC parity-check matrices to adapt the classical Vontobel--Smarandache construction \cite{VS} in a computationally efficient way. This yields solid upper bounds on the minimum distance of quasi-cyclic 5G LDPC codes. For obtaining upper bounds, by comparison, the Brouwer--Zimmermann algorithm implemented in MAGMA \cite{BC} is slow for long codes, while the number-geometry-based estimation method of \cite{U} via the open-source MATLAB code of \cite{U3}, is limited in the input code lengths. In particular, the MATLAB implementation in \cite{U3} works for lengths up to 8000 bits, whereas BG1 5G codes have lengths up to $25344$. 

In Section~\ref{Sec_reduced} we develop a method for bounding the minimum distances of quasi-cyclic codes based on the natural relationship between a quasi-cyclic code with circulant size $q$ and its modular reduction with circulant size $q'$, where $q' \mid q$. This method turns out to be particularly useful for deriving lower bounds and, when combined with the upper bounds from Section~\ref{Sec_Vont}, provides tight bounds for minimum distances of several 5G codes with lengths up to $4992$.

In addition to the minimum distance analysis, we further exploit the ideas from Section~\ref{Sec_reduced} and study a practical aspect of early decoding termination in Section ~\ref{Sec_ET}. The LDPC decoder (typically a min-sum decoder) approximates belief propagation by passing messages between variable and check nodes. It computes soft updates using the min-sum rule to reduce complexity and employs early termination if the syndrome vector is zero (indicating successful decoding) or after a maximum number of iterations. There are also alternative approaches to early termination for quasi-cyclic LDPC codes; see, e.g., \cite{OB}.

In the classic column-driven min-sum approach to syndrome calculation, at the end of each iteration, the impact of each circulant block must be computed using the corresponding shifter area. With this approach, the number of shifters in hardware equals the maximum column weight of the parity-check matrix.

The hardware for 5G codes must support all codes defined in the standard, as well as the so-called all-layer code, which has the largest parity-check matrix. The first difficulty arises from the fact that the parity-check matrix of the all-layer code has the first two circulant columns dense providing a maximum column weight of 
$30$, see Fig. \ref{fig:bg1-bg2-mask}, thus requiring 
$30$ shifters in the area. This is unusually large for standard LDPC codes; for example, LDPC codes from IEEE 802.11ax WiFi standard  have a maximum column weight of only up to $11$ and DVB-S2 codes have a maximum column weight of only $6$. 

Moreover,  the LDPC codes in the 5G standard cover a huge range of circulant sizes (51 possible values for shifters) \cite{5G}. The traditional approach does not replace each shifter with $51$ individual shifters but rather implements several {\it  extended shifters} (as they were called in \cite {BoutillonHarb2020})—pieces of hardware that shift any data of size up to the maximum possible. The traditional implementation \cite{XBH} of such single module uses  two shifters, which doubles the area per extended shifter. Note that other approaches for implementation of extended shifters have been developed, and    for 5G codes, in particular,  \cite{BoutillonHarb2020}, \cite{Zhong2020}.

In Section~\ref{Sec_ET} we propose a new early termination approach based on circulant modular reduction of the syndrome, which reduces the computational complexity of syndrome calculation and can be useful in hardware implementations of 5G LDPC decoders. Thus, the results of the paper are relevant both from theoretical standpoints (minimum distance analysis) and from engineering standpoints (efficient decoder design).
\begin{figure}[h]
    
    \begin{tikzpicture}
\begin{semilogyaxis}[
    width=0.86\textwidth,
    height=0.50\textwidth,
    xmin=2.85,
    xmax=4.15,
    ymin=5e-10,
    ymax=2,
    xtick={2.9,3.0,3.1,3.2,3.3,3.4,3.5,3.6,3.7,3.8,3.9,4.0,4.1},
    xticklabels={$2.9$,$3.0$,$3.1$,$3.2$,$3.3$,$3.4$,$3.5$,$3.6$,$3.7$,$3.8$,$3.9$,$4.0$,$4.1$},
    ytick={1,1e-2,1e-4,1e-6,1e-8},
    yticklabels={$1$,$10^{-2}$,$10^{-4}$,$10^{-6}$,$10^{-8}$},
    minor y tick num=1,
    tick align=outside,
    tick style={black},
    xlabel={SNR (dB)},
    ylabel={Information block error rate},
    xlabel style={font=\small},
    ylabel style={font=\small},
    tick label style={font=\small},
    legend style={
        font=\tiny,
        at={(0.99,0.98)},
        anchor=north east,
        draw=none,
        fill=white,
        fill opacity=0.85,
        text opacity=1
    },
    ymajorgrids,
    xmajorgrids,
    grid style={gray!25},
    axis line style={black},
]
\addplot[
      thick,
    mark=o,
    mark size=1.9pt,
    color=blue!70,
] coordinates {
    (2.9,3.8e-1)
    (3.0,9.5e-2)
    (3.1,1.1e-2)
    (3.2,6.0e-4)
    (3.3,1.5e-5)
    (3.4,1.2e-6)
    (3.5,5.0e-7)
    (3.6,2.5e-7)
    (3.7,1.1e-7)
    (3.8,1.1e-7)
    (3.9,8.0e-8)
    (4.0,2.0e-8)
    (4.1,1.0e-8)
};
\addlegendentry{5G code (1)}

\addplot[
    thick,
    mark=*,
    mark size=1.7pt,
    color=red!75!black,
] coordinates {
    (3.0,8.9e-1)
    (3.2,2.5e-1)
    (3.4,6.0e-3)
    (3.6,7.0e-6)
    (3.7,8.0e-8)
    (3.75,2.0e-9)
};
\addlegendentry{random (2)}

\end{semilogyaxis}
\end{tikzpicture}
    \caption{Information block error rate for 5G LDPC code vs same rate and information block length random code}\label{fig:ibler_intro}
\end{figure}
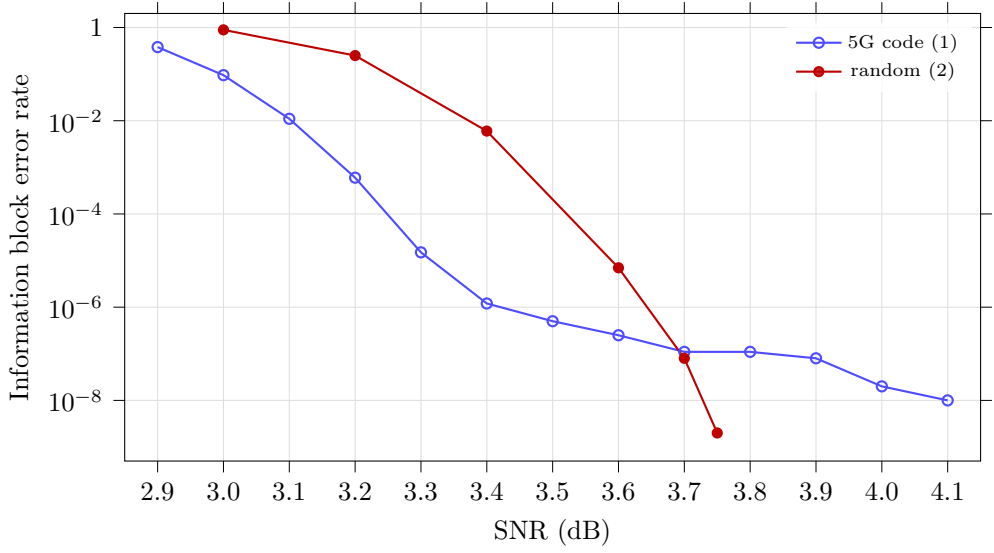 

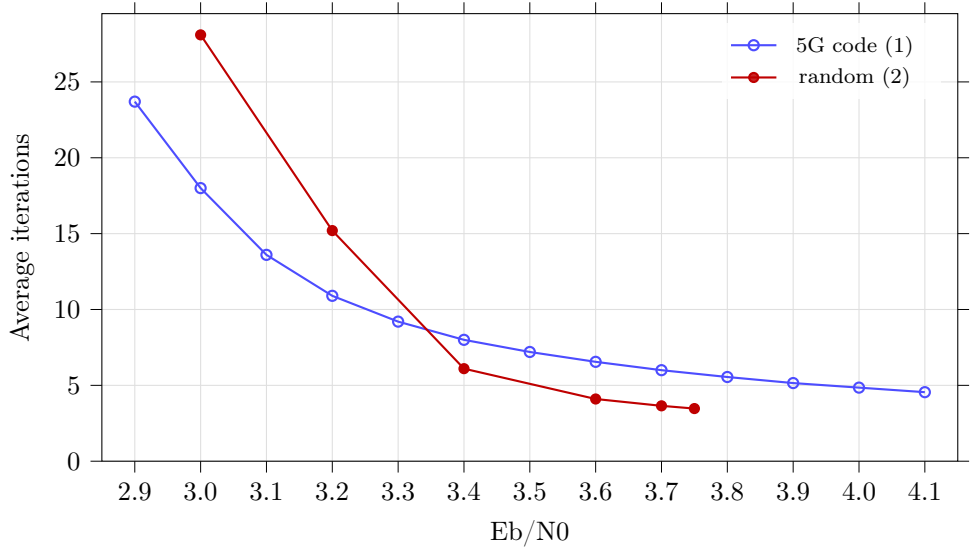
\begin{figure}[h]
    
    \begin{tikzpicture}
\begin{axis}[
    width=0.86\textwidth,
    height=0.50\textwidth,
    xmin=2.85,
    xmax=4.15,
    ymin=0,
    ymax=29.5,
    xtick={2.9,3.0,3.1,3.2,3.3,3.4,3.5,3.6,3.7,3.8,3.9,4.0,4.1},
    xticklabels={$2.9$,$3.0$,$3.1$,$3.2$,$3.3$,$3.4$,$3.5$,$3.6$,$3.7$,$3.8$,$3.9$,$4.0$,$4.1$},
    ytick={0,5,10,15,20,25},
    tick align=outside,
    tick style={black},
    xlabel={Eb/N0},
    ylabel={Average iterations},
    title style={font=\small, yshift=-0.5ex},
    xlabel style={font=\small},
    ylabel style={font=\small},
    tick label style={font=\small},
    legend style={
        font=\scriptsize,
        at={(0.98,0.98)},
        anchor=north east,
        draw=none,
        fill=white,
        fill opacity=0.85,
        text opacity=1
    },
    ymajorgrids,
    xmajorgrids,
    grid style={gray!25},
    axis line style={black},
]

\addlegendentry{\begin{tabular}{l}
    5G code (1)
\end{tabular}}
\addplot[
    thick,
    mark=o,
    mark size=1.9pt,
    color=blue!70,
] coordinates {
    (2.9,23.7)
    (3.0,18.0)
    (3.1,13.6)
    (3.2,10.9)
    (3.3,9.2)
    (3.4,8.0)
    (3.5,7.2)
    (3.6,6.55)
    (3.7,6.0)
    (3.8,5.55)
    (3.9,5.15)
    (4.0,4.85)
    (4.1,4.55)
};

\addlegendentry{random (2)}

\addplot[
    thick,
    mark=*,
    mark size=1.7pt,
    color=red!75!black,
] coordinates {
    (3.0,28.1)
    (3.2,15.2)
    (3.4,6.1)
    (3.6,4.1)
    (3.7,3.65)
    (3.75,3.47)
};

\end{axis}
\end{tikzpicture}
    \caption{Average iterations for 5G LDPC code vs same rate and information block length random code} \label{fig:average_iterations_intro}
\end{figure}

{\bf Observed Error-Floor Behavior for 5G LDPC codes.} The tests in this work were performed using the floating-point layered min-sum decoder and the following 5G LDPC code. In accordance with \cite{5G}, we consider a left-punctured code formed by the first six rows of base graph 1 with lifting set 1 exponent matrix and has circulant size $384$. The code has the following parameters
\begin{equation}\label{5G code}\mbox{5G LDPC code: length 9984, information block length 8448}\end{equation}
and it is the code $C'({\mathcal E}_{[6]},384)$ in the notation of Section~\ref{subsec_codes}. As one of the results that we show in this study, the minimum distance of the code (\ref{5G code}) is at most $14$ and at least $8$.

In order to compare the performance we chose a random quasi-cyclic LDPC code with circulant size $192$, fixed column weight four in its parity-check matrix and  parameters:\begin{equation}\label{Random code}\mbox{Random code: length 9984, information block length 8447},\end{equation}
so  having almost the same rate and message length as the 5G code (\ref{5G code}).

In the high SNR region, the 5G LDPC code exhibits a distinctive error floor for the information block, see Fig.~\ref{fig:ibler_intro}, which is absent for a random code whose error rates remain relatively similar at an $Eb/N0$ of around $3.7$ dB. From the average iterations standpoint, the 5G LDPC code starts to perform worse than the random code over a broader SNR region (for $Eb/N0 \geq 3.4$ dB), see Fig.~\ref{fig:average_iterations_intro}.

The tests show that, in particular, the BLER of 5G codes is affected by the relatively small minimum distance, as there are instances where the LDPC decoder outputs undetected errors, i.e., events where the decoder converges to a codeword that is different from the transmitted one, see Fig.~\ref{FigET_UIBLER}. This observation underlines the importance of the 24-bit CRC employed in the 5G standard for catching such undetected error events \cite{Bai}.

\section{Preliminaries}

\subsection{Basic definitions}\label{subsec21}

Let $\mathbb{F}_2^n$ denote the $n$-dimensional vector space over the binary field $\mathbb{F}_2$.
{\it A binary linear code} of length $n$ and dimension~$k$ is a~$k$-dimensional
subspace of $\mathbb{F}_2^n$. {\it The  Hamming distance} between words is equal to the number of distinct positions. {\it The weight} of a word $c$ is defined as the number of nonzero symbols in $c$ and denoted by $w(c)$. {\it The minimum distance} of a linear code is the minimum
Hamming distance between two distinct codewords, or equivalently, the minimum
weight of its nonzero codewords. For a linear code $C$, we denote its minimum
distance by $d(C)$.

A {\it generator matrix} of a linear code is a matrix whose rows form a basis of the
code. Given a matrix $M$ and a set $S$ of column indices, we denote by $M^S$
the submatrix of $M$ formed by the columns indexed by $S$. A set $I$ of coordinate
positions of a linear code $C$ with generator matrix $G$ is called {\it an
information set} if $G^I$ is nonsingular.

A matrix $H$ is called {\it a  parity-check matrix} of a linear code $C$ if
\[
    Hc^T = 0 \Leftrightarrow c \in C.
\]

{\it The concatenation} of two vectors $u$ and $v$ is denoted by $(u,v)$.

Let $I$ be a set of position indices of a code $C$. {\it The punctured code} $C'$ is
obtained from $C$ by deleting the symbols in the codewords at positions in~$I$.
If the sizes of the code $C$ and its punctured code $C'$ coincide, we call $C'$
{\it an iso-punctured code}. We note the following simple property for an iso-punctured code $C'$:
\begin{equation}
    \label{punc_min_dist}
    d(C') \leq d(C).
\end{equation}

\textbf{Remark 1 (iso-punctured 5G LDPC codes).} If $H$ is a parity-check matrix
of a linear code $C$ and $I$ is a set of positions such that $H^I$ is a full-rank matrix,
then the code $C'$ obtained  from $C$ by puncturing in the positions from $I$ is iso-punctured. In particular,
all punctured codes used in the 5G standard are obtained from certain codes
(all-layer codes) and are iso-punctured,  see Sections~2.3--2.4 below for more details.

\subsection{Quasi-cyclic codes by modular lifting}

A linear code $C$ is called {\it quasi-cyclic} if there exists a divisor $q$ of its
length such that
\[
    (c_1,\ldots,c_n)\in C \Leftrightarrow
    (c_q,c_1,\ldots,c_{q-1},\ldots,c_n,c_{n-q+1},\ldots,c_{n-1})\in C.
\]
In particular, quasi-cyclic codes with $q=n$ are exactly cyclic codes.

Let $E$ be a matrix whose elements are integers greater than or equal to $-1$.
We denote by $C(E,q)$ the binary linear code with the parity-check
matrix obtained from $E$ as follows:

\begin{itemize}
    \item each element $E_{ij}\neq -1$ is replaced by the circulant matrix of
          order $q$ corresponding to a cyclic shift by $E_{ij} \bmod q$ positions;
    \item each element equal to $-1$ is replaced by the zero matrix of order
          $q$.
\end{itemize}

By construction, the code $C(E,q)$ is quasi-cyclic. The matrix obtained from
$E$ by reducing every element different from $-1$ modulo $q$ is called 
    {\it the exponent matrix} of $C(E,q)$. The number $q$ is called 
    {\it the circulant size} of the code.

We index codeword positions starting from $1$. Let $c$ be a codeword of a
quasi-cyclic code $C$ of length $qn$ and circulant size $q$. The
    {\it block support} of $c$ is the set
\[
    \{\lceil \tfrac{i}{q} \rceil \,:\, c_i=1,\ i=1,\ldots,qn\}.
\]

\subsection{5G LDPC codes}\label{subsec_codes}
To support a wide range of code rates and block lengths, the 5G standard uses
a family of quasi-cyclic LDPC codes \cite{5G}. These codes are based on two
collections of exponent matrices:

\begin{itemize}
    \item eight matrices of size $46 \times 68$, called {\it base graph 1},
          or briefly BG1;
    \item eight matrices of size $42 \times 52$, called {\it base graph 2},
          or briefly BG2.
\end{itemize}

The matrices from the two base graphs are paired so that each pair has the
same circulant size; we call this size the {\it parent circulant size}. The set of all parent circulant sizes is  
\[
    2^{8}, \quad 3\cdot  2^{7}, \quad 5 \cdot 2^{6}, \quad 7 \cdot 2^{5},
    \quad 9 \cdot 2^{5}, \quad 11 \cdot 2^{5}, \quad 13 \cdot 2^{4},
    \quad 15 \cdot 2^{4}.
\]
All circulant sizes $q$ used in the standard are related to the corresponding
parent circulant size $\overline{q}$ by
\[
    \frac{\overline{q}}{q}=2^i, \mbox{ for all }  i\geq 0.
\]

For BG1 (respectively BG2) and any parent circulant size $\overline{q}$, we
denote the corresponding $46\times 68$ (respectively $42\times 52$) exponent
matrix by $\mathcal{E}$ and call it the {\it all-layer exponent matrix}. For
any $q$ such that $\overline q=2^i q$, the code $C(\mathcal{E},q)$ is called
an {\it all-layer code}.

An important feature of the 5G LDPC design is the following block structure of
the exponent matrix $\mathcal{E}$ and the corresponding parity-check matrix of
the all-layer code:
\begin{equation}
    \label{decomp}
    \mathcal{E}=\left(%
    \begin{array}{cccc} \epsilon & -1 \\
                 \mathcal{E}'    & {\mathcal T}
                 \\\end{array}%
    \right)\leftrightarrow
    H= \left(%
    \begin{array}{cccc} h & {\bf 0} \\
                 H'       & I
                 \\\end{array}%
    \right),
\end{equation}
where  $\epsilon$ is a
$4\times 26$ matrix, ${\mathcal T}$ is a $42\times 42$ matrix with zeros on the main diagonal and $-1$ elsewhere,   $I$ is the identity matrix of size $42q\times 42q$. For BG2 the structure is analogous, with an identity 
matrix $I$ of size $38q\times 38q$ and a $4\times 14$ matrix $\epsilon$.

For $j\geq 4$ and an all-layer matrix $\mathcal E$, we denote by
${\mathcal E}_{[j]}$ its upper-left submatrix of size
\[
    j\times (22+j) \quad \text{for BG1},
\]
and
\[
    j\times (10+j) \quad \text{for BG2}.
\]

In the 5G standard, the transmitted messages are obtained from codewords of
$C({\mathcal E}_{[j]},q)$ by puncturing the leftmost $2q$ bits. We denote the
resulting punctured code by $C'({\mathcal E}_{[j]},q)$.

Thus, the LDPC codes supported by the standard can be summarized as follows:
\begin{equation}
    \label{eq:all_layer_codes}
    \begin{aligned}
         & \text{For any of the $16$ all-layer exponent matrices $\mathcal{E}$}                       \\
         & \text{with fixed parent circulant size $\overline{q}$, the codes } C'(\mathcal{E}_{[j]},q) \\
         & \text{are considered for all $j\geq 4$ and all $q$ such that } \overline{q}=q2^i,\ i\geq 0.
    \end{aligned}
\end{equation}

Let us consider decomposition (\ref{decomp}). We observe that the $4\times 26$ matrix $\epsilon$ is precisely ${\mathcal E}_{[4]}$, and we rewrite (\ref{decomp}) as follows:

\begin{equation}
    \label{eq:decomp2}
    \mathcal{E} = \left(
    \begin{array}{cccc}  
        {\mathcal E}_{[4]} & -1 \\
        \mathcal{E}'       & {\mathcal T}\\
    \end{array}
    \right)
    \leftrightarrow
    H = \left(
    \begin{array}{cccc} 
        h  & 0 \\
        H' & I \\
    \end{array}
    \right).
\end{equation}

Throughout the work, we are especially focused on the codes $C({\mathcal E}_{[4]},q)$ 
and $C({\mathcal E}_{[6]},q)$ and their left puncturings $C'({\mathcal E}_{[4]},q)$ 
and $C'({\mathcal E}_{[6]},q)$, which we call {\it $4$-layer} and {\it $6$-layer codes}.

\subsection{Minimum distance properties of punctured 5G LDPC codes}

In this paper, we study the minimum distances of the codes $C({\mathcal E}_{[j]},q)$ 
and their left puncturings $C'({\mathcal E}_{[j]},q)$. In terms of error-correction 
capability, both classes are tightly connected due to the reasoning below.

\begin{itemize}[leftmargin=0pt,style=unboxed]
    \item In order to reconstruct the punctured bits, the 5G LDPC decoder has to operate with the parity-check matrix of the
        nonpunctured code $C({\mathcal E},q)$. As a result, its behavior is more
        naturally related to the code space of $C({\mathcal E},q)$ than to that of
        $C'({\mathcal E},q)$.

    \item The code $C'({\mathcal E}_{[j]},q)$ is 
        an iso-punctured code of the all-layer code $C({\mathcal E},q)$.    
        The nonpunctured code $C({\mathcal E}_{[j]},q)$ obviously has minimum
        distance at least as large as that of the punctured code
        $C'({\mathcal E}_{[j]},q)$ and for $4\leq i\leq j$
        \begin{equation}
            \label{eq:left_punc}
            d(C'({\mathcal E}_{[i]},q)) \leq 
            d(C'({\mathcal E}_{[j]},q)) \leq 
            d(C({\mathcal E}_{[j]},q))  \leq 
            d(C({\mathcal E},q)).
        \end{equation}
        Indeed, for $j\geq 4$, the first $2q$ and the last $42q$ BG1 ($38q$ for BG2) 
        columns of the parity-check matrix obtained from $\mathcal{E}$ are linearly 
        independent, which implies (\ref{eq:left_punc}) by Remark 1. 

        Moreover, the leftmost $2q$ bits are "highly protected"\, because the
        corresponding first columns of the parity-check matrix have large weight, see Fig.~\ref{fig:bg1-bg2-mask}.
        Therefore, for sufficiently large circulant sizes $q$, minimum weight and
        small-weight codewords of $C({\mathcal E},q)$ are unlikely to contain nonzero
        bits in the punctured positions. This is consistent with our observations for
        moderate and large code lengths; see Tables~\ref{Table_1} and~\ref{Table_p} for the case
        $q=24$ and the results of the study of the minimum distances represented in Table~\ref{T3}.
\end{itemize}

\begin{figure}[H]
     \includegraphics[width=1.0\textwidth]{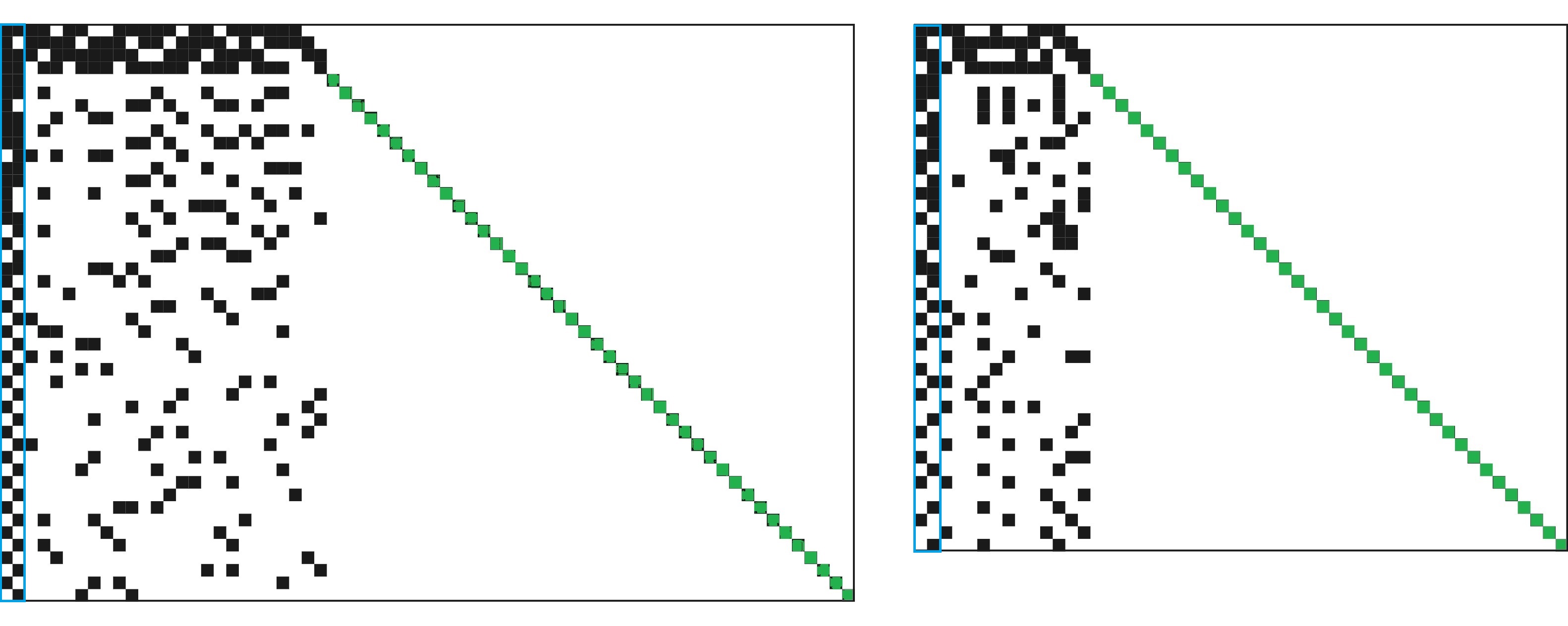}
    \caption{The BG1 and BG2 exponent matrices $\mathcal{E}$. $\blacksquare$ if ${\mathcal E}_{ij} \not= -1$,  $\square$ if ${\mathcal E}_{ij} = -1$ and $\greensquare$ if ${\mathcal E}_{ij} = 0$ (correspond to identity circulants) at the bottom right, $\bluesquare$ are the first two columns which are moderately dense.}  \label{fig:bg1-bg2-mask}
   
\end{figure}

\subsection{Decoding  of 5G LDPC codes}

According to the 5G standard \cite{5G}, at the encoder side the information bits are followed by a CRC block, forming the LDPC information block; the encoder then calculates the LDPC parity block, see Fig. ~\ref{fig:5g-codeword}.  

At the decoding end, a  high-level  scheme for 5G LDPC codes is illustrated in Fig.~\ref{5g:scheme}.
Each LDPC decoding attempt either terminates early when the full (or partial)
syndrome condition is satisfied, or reaches the maximum number of iterations.
In the latter case, the decoder still outputs an LDPC information block, which is
then passed to the CRC stage for additional error detection. In what follows,
we are mainly interested in the error-correction capability of the LDPC block
itself, independently of the final CRC verification.

We define the LDPC {\it information block error rate} (IBLER) as the ratio of the
number of frames in which the decoded information block differs from the
transmitted one to the total number of frames. Similarly, we define the
{\it undetected information block error rate} (UIBLER) for the LDPC block output
as the ratio of the number of frames in which the decoded information block is
incorrect and the early termination condition is  satisfied, to the
total number of frames. 

For a
$24$-bit CRC defined in the standard for LDPC codes, the corresponding
combined  LDPC and CRC undetected information block error rate is expected
to be around $2^{-24}\cdot 10^{-7}$ for a floating point decoder. Here the term
$10^{-7}$ comes from the worst-case LDPC UIBLER, see Section~\ref{Sec_ET} for
more details. This extremely low error probability, which cannot be captured
with ordinary computer simulations, is one of the reasons CRC-related
considerations are excluded from the current study.

\begin{figure}[H]
    
    \begin{tikzpicture}[x=1cm,y=1cm,font=\small]
    \def\h{1.25}
    \def\payloadW{6.3}
    \def\infoW{7.8}
    \def\totalW{11.8}

    \draw[thick] (0,0) rectangle (\payloadW,\h);
    \draw[thick,fill=black!20] (\payloadW,0) rectangle (\infoW,\h);
    \draw[thick,fill=black!40] (\infoW,0) rectangle (\totalW,\h);
    \draw[thick] (0,0) rectangle (\totalW,\h);
    \draw[thick] (\payloadW,0) -- (\payloadW,\h);
    \draw[thick] (\infoW,0) -- (\infoW,\h);

    \node at (3.15,0.63) {information bits};
    \node at (7.05,0.63) {CRC};
    \node at (9.8,0.63) {parity bits};

    \draw[decorate,decoration={brace,amplitude=4pt}]
        (0,1.50) -- (\infoW,1.50)
        node[midway,yshift=0.32cm] {LDPC information block};

    \draw[decorate,decoration={brace,amplitude=4pt}]
        (\infoW,1.50) -- (\totalW,1.50)
        node[midway,yshift=0.32cm] {LDPC parity block};

    \draw[decorate,decoration={brace,mirror,amplitude=4pt}]
        (0,-0.20) -- (\payloadW,-0.20);
    \draw[decorate,decoration={brace,mirror,amplitude=4pt}]
        (\payloadW,-0.20) -- (\infoW,-0.20);

    \draw[-{Triangle[length=3mm,width=2.2mm]},thick]
        (3.15,-0.31) -- (3.15,-0.72)
        -- node[midway,below] {CRC encoder} (7.05,-0.72)
        -- (7.05,-0.31);

    \draw[decorate,decoration={brace,mirror,amplitude=4pt}]
        (0,-1.30) -- (\infoW,-1.30);
    \draw[decorate,decoration={brace,mirror,amplitude=4pt}]
        (\infoW,-1.30) -- (\totalW,-1.30);

    \draw[-{Triangle[length=3mm,width=2.2mm]},thick]
        (3.90,-1.41) -- (3.90,-1.83)
        -- node[midway,below] {LDPC encoder} (9.80,-1.83)
        -- (9.80,-1.41);
\end{tikzpicture}
    \caption{Schematic structure of a 5G LDPC codeword}
    \label{fig:5g-codeword}
\end{figure}
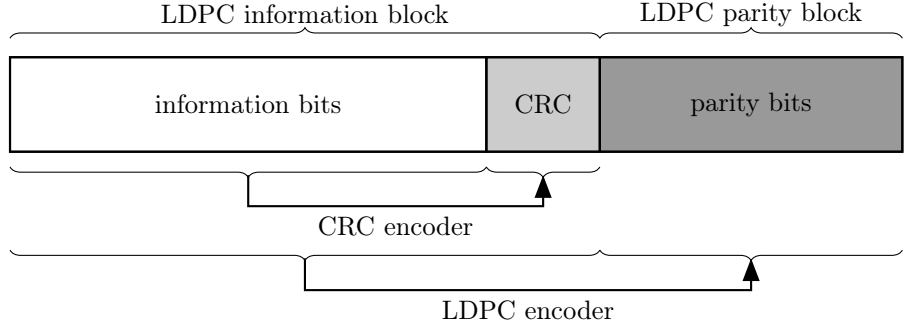

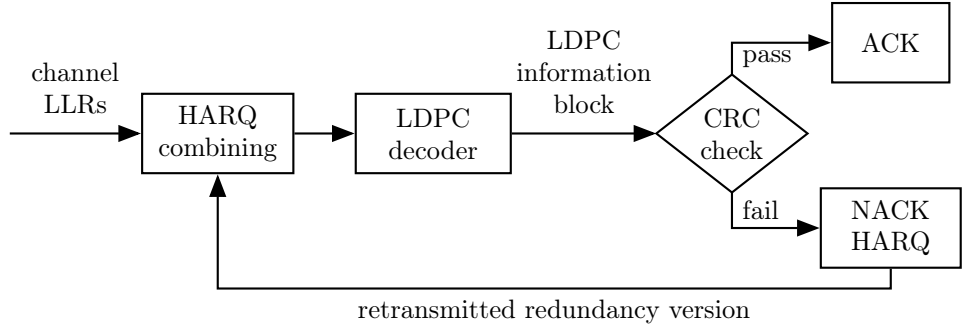
\begin{figure}[H]
    
    \begin{tikzpicture}[x=1cm,y=1cm,font=\small]
    \tikzset{
        block/.style={draw,thick,minimum height=1.05cm,align=center},
        arrow/.style={-{Triangle[length=3mm,width=2.2mm]},thick}
    }

    \node[block,minimum width=2.00cm] (harqbuf) at (2.75,0) {HARQ\\combining};
    \node[block,minimum width=2.05cm] (decoder) at (5.60,0) {LDPC\\decoder};

    \coordinate (crcTop) at (9.55,0.78);
    \coordinate (crcRight) at (10.55,0);
    \coordinate (crcBottom) at (9.55,-0.78);
    \coordinate (crcLeft) at (8.55,0);
    \draw[thick] (crcTop) -- (crcRight) -- (crcBottom) -- (crcLeft) -- cycle;
    \node[align=center] at (9.55,0) {CRC\\check};

    \node[block,minimum width=1.55cm] (ack) at (11.65,1.20) {ACK};
    \node[block,minimum width=1.85cm] (nack) at (11.65,-1.25) {NACK\\HARQ};

    \draw[arrow] (0,0) -- node[midway,above] {\begin{tabular}{c}channel\\LLRs\end{tabular}} (harqbuf.west);

    \draw[arrow] (harqbuf.east) -- (decoder.west);
    \draw[arrow] (decoder.east) -- node[midway,above] {\begin{tabular}{c}LDPC\\information\\block\end{tabular}} (crcLeft);

    \draw[arrow] (crcTop) |- node[pos=0.25,right] {pass} (ack.west);
    \draw[arrow] (crcBottom) |- node[pos=0.25,right] {fail} (nack.west);

    \draw[arrow] (nack.south) -- (11.65,-2.05)
        -- node[midway,below] {retransmitted redundancy version} (2.75,-2.05)
        -- (harqbuf.south);
\end{tikzpicture}
    \caption{5G LDPC decoding scheme}
    \label{5g:scheme}
\end{figure}

\subsection{Optical Communications Terminal standard codes}
 
Some of the 5G LDPC codes introduced in \cite{5G} were later adopted in other
communication standards, for example in the Optical Communications Terminal
Standard Version 4.0.0 \cite{SDA}. This standard uses the following four BG1
codes \cite[Section 3.4.6.1.3]{SDA}:
\[
    C'({\mathcal E}_{[6]},384), \quad C'({\mathcal E}_{[9]},384), \quad
    C'({\mathcal E}_{[13]},384), \quad C'({\mathcal E}_{[24]},384).
\]
In the OCT standard, each information block contains a $144$-bit header
followed by a $16$-bit CRC, and the block is additionally protected by $800$
parity bits of a convolutional code \cite[Section 3.4.5.2]{SDA}. As a result,
the information block error-detection capabilities of this standard are much
stronger than in 5G, partly because of a simpler retransmission ARQ protocol.
By contrast, 5G appends a $24$-bit CRC to the information block and uses a
more sophisticated retransmission HARQ protocol.

\section{Vontobel-Smarandache lemma for 5G LDPC codes}\label{Sec_Vont} 

\subsection{Vontobel-Smarandache lemma}

Let $E$ be an $r\times n$  matrix with integer entries greater than or 
equal to $-1$. For a positive integer $q$, define the matrix
$P(E,q)$ over the quotient ring
\[
    \mathbb{F}_2[x]/(x^q+1)
\]
by
\begin{equation}
\label{eq:polynomial-shift-matrix}
    P(E, q)_{ij} = 
    \begin{cases} 
        0,                 & E_{ij}   = -1,\\
        x^{E_{ij}\bmod q}, & E_{ij}\neq -1.
    \end{cases}
\end{equation}
We call $P(E,q)$ {\it the polynomial shift matrix} associated with the
quasi-cyclic code $C(E,q)$. For $a(x)=\sum_{t=0}^{q-1}a_t x^t\in \mathbb{F}_2[x]/(x^q+1)$, let
\[
    \overline{a}=(a_0,\ldots,a_{q-1}).
\]
be its coefficient vector. Following Section \ref{subsec21}, for a set $S$ of column indices of $P$, $P^S$ denotes the
submatrix of $P$ formed by the columns indexed by $S$. For a square matrix $M$, we denote its determinant by $det(M)$. The following theorem gives the  construction of
Vontobel and Smarandache.
\begin{theorem}\cite{VS}\label{Th_VS}    Let $P=P(E,q)$ be the $r\times n$ polynomial shift matrix of the code $C(E,q)$, and let
    $S=\{s_1,\ldots,s_{r+1}\}\subseteq\{1,\ldots,n\}$.
    Define a block vector $c(S)=(c_1,\ldots,c_n)$, where each block $c_j$ has length
    $q$, by 
    \[
        c_j =
        \begin{cases}
            \overline{\det\left(P^{S\setminus\{j\}}\right)}, & j\in S,\\
            \overline{0},                                    & j\notin S.
        \end{cases}
    \]
    Then $c(S)$ is a codeword of $C(E,q)$.
\end{theorem}


Whenever $c(S)$ is nonzero, its Hamming weight gives an upper bound
on $d(C(E,q))$. The method is most computationally effective for small number $r$, which is the number  of rows $P$, since it requires
determinants of $r\times r$ matrices over $R_q$. It has been
applied to several code families, including AR4JA codes \cite{ButlerSiegel}.
In \cite{Butler}, Butler used it to bound the minimum distances of several
IEEE~802 communication standard codes with $r\leq 12$.

\subsection{Limitations for the all-layer 5G code} 

We first consider the all-layer BG1 code $C({\mathcal E},384)$, where
$\mathcal{E}$ is the $46\times68$ exponent matrix of lifting set~1 (LS1) in
the 5G NR standard \cite{5G}. A direct application of
Theorem~\ref{Th_VS} to this code is not computationally effective for the reasons below. 

\begin{itemize}[leftmargin=0pt,style=unboxed]
    \item For each subset $S$ of size $47$, the construction requires
          $47$ determinants of $46\times46$ matrices over the ring $\mathbb{F}_2[x]/(x^{384}+1)$. In our
          experiments, the fastest  
          implementation in Magma \cite{BC} required about $10$ minutes per codeword. Our C++ 
          implementation and publicly available implementations such as~\cite{U2} were slower.

    \item Exhaustive enumeration is infeasible: one would have to inspect
          $\binom{68}{47}$ subsets. Randomly selected subsets produce
          codewords of moderate or large weight; after several days of search,
          the smallest observed weight was $456$, while typical weights were
          in the thousands.
\end{itemize}

Thus, although Theorem~\ref{Th_VS} is a powerful source of explicit
codewords, applying it directly to the  all-layer exponent matrix with $46$ rows does
not yield useful minimum distance bounds for the longest 5G code.

\subsection{Vontobel--Smarandache lemma via the 4-layer code} 
\label{subsec:sec_VS_4lay}

The block decomposition in \eqref{eq:decomp2} provides a substantially more
efficient way to use the same determinant construction. Recall that, for BG1,
\[
    \mathcal{E}=
    \begin{pmatrix}
        {\mathcal E}_{[4]} & -1\\
        \mathcal{E}'       & {\mathcal T}
    \end{pmatrix},
    \qquad
    H=
    \begin{pmatrix}
        h  & 0\\
        H' & I
    \end{pmatrix},
\]
where ${\mathcal E}_{[4]}$ is a $4\times26$ exponent matrix and $h$ is the
corresponding parity-check matrix. Applying Theorem~\ref{Th_VS} to
$C({\mathcal E}_{[4]},384)$ requires only the computation of determinants of
$4\times4$ matrices over $\mathbb{F}_2[x]/(x^{384}+1)$ and the enumeration of
$\binom{26}{5}$ subsets. This computation is practical on a personal computer.

The key observation is that every codeword $c$ of the 4-layer code (i.e., $hc^T=0$) can be
extended to a codeword of the all-layer code with parity-check matrix $H$ as above by appending the uniquely
determined parity part:
\[
    c^\star=(c,cH'^T).
\]
Since $H'$ is sparse, computing $cH'^T$ is negligible
compared with direct determinant evaluation via the Vontobel--Smarandache lemma for the all-layer matrix.

\begin{construction}\label{Cons_1} (Vontobel--Smarandache upper bound via the 4-layer code)
    Fix a circulant size $q$ and the decomposition~\eqref{eq:decomp2}.
    \begin{enumerate}
        \item For each $5$-subset $S'\subseteq\{1,\ldots,26\}$, use
              Theorem~\ref{Th_VS} to construct the codeword
              $c(S')\in C({\mathcal E}_{[4]},q)$.
        \item If $c(S')\neq 0$, form the all-layer codeword
              \[
                  c^\star(S')=(c(S'),c(S')H'^T)\in C({\mathcal E},q).
              \]
    \end{enumerate}
\end{construction}    The resulting collection of weights gives explicit  bounds on the minimum distance
    of the all-layer code. We illustrate the approach with the following toy example.
\begin{example}\label{ex:vs-polynomial-minors-ext}
Let  the exponent matrix $E$ for quasi-cyclic LDPC code with $q=2$ be 
    $$       E=
    \left(\begin{array}{c|c}
    {E}_{[4]} & -1\\
    \hline
        E'       & {\mathcal T}    
\end{array}\right)
=$$
$$\left(\begin{array}{ccccccccc|ccc}
  -1& 0& 1& 1& 1& 1& 0& -1& -1& -1& -1& -1 \\
     1& -1& 1& -1& 0& 0& 0& 0& -1& -1& -1& -1 \\
     1& 1& 1& 1& -1& -1& -1& 0& 0& -1& -1& -1 \\
     0& 0& -1& 0& 0& 1& -1& -1& 0& -1& -1& -1 \\
     \hline
     -1& -1& -1& -1& -1& -1& -1& -1& -1& 0& -1& -1 \\
     -1& -1& -1& -1& 1& 1& -1& -1& -1& -1& 0& -1 \\
     1& 1& -1& 1& -1& -1& -1& -1& -1& -1& -1& 0 
\end{array}\right).
$$ It has similar decomposition structure as 5G exponent matrix (\ref{eq:decomp2}) with the $3\times 3$ submatrix ${\mathcal T}$, composed of $-1$ and $0$ on the  rightmost bottom side. 
 We apply 
    Theorem~\ref{Th_VS} for the subcode with "toy $4$-layers"\, exponent matrix         ${E}_{[4]}$ and $S'=\{1,2,3,6,8\}, \{1,3,4,7,8\}$
to produce two weight three codewords of $C(E_{[4]},2)$, which we represent over the  ring $F_2[x]/(x^2+1)$: 
    \[
        c(\{1,2,3,6,8\})=\bigl(0,x,0,0,0,1,0,1,0\bigr),
        c(\{1,3,4,7,8\})=\bigl(x,0,0,x,0,0,1,0,0\bigr).
    \]
    Further, these codewords of the code $C(E_{[4]},2)$ provide codewords of the code $C(E,2)$ as described above, and the additional parity bits are found as: $$  P'\bigl(0,x,0,0,0,1,0,1,0\bigr)^T=\begin{pmatrix}
0\\
x\\
1\\
\end{pmatrix}, P'\bigl(x,0,0,x,0,0,1,0,0\bigr)^T=\begin{pmatrix}
0\\
0\\
0\\
\end{pmatrix},$$ where  $P'=\begin{pmatrix}
0,0,0,0,0,0,0,0,0\\
0,0,0,0,x,x,0,0,0\\
x,x,0,x,0,0,0,0,0\\
\end{pmatrix}$. We obtain low-weight words of $C(E,2)$:

\small$$\bigl(\overline{0},\overline{x},\overline{0},\overline{0},\overline{0},\overline{1},\overline{0},\overline{1},\overline{0},\overline{0},\overline{x},\overline{1}\bigr)= \bigl(0,0,0,1,0,0,0,0,0,0,1,0,0,0,1,0,0,0,0,0,0,1,1,0\bigr),$$\small$$\bigl(\overline{x},\overline{0},\overline{0},\overline{x},\overline{0},\overline{0},\overline{1},\overline{0},\overline{0},\overline{0},\overline{0},\overline{0}\bigr)=\bigl(0,1,0,0,0,0,0,1,0,0,0,0,1,0,0,0,0,0,0,0,0,0,0,0\bigr). $$


\end{example}

For $q=384$, each nonzero codeword obtained from Theorem~\ref{Th_VS} generates
an orbit of size $384$ under simultaneous cyclic shifts of all circulant
blocks. The resulting weight distribution is shown in
Fig.~\ref{fig:5g-weights}.

\begin{figure}[H]

    \begin{tikzpicture}
\begin{axis}[
    width=0.86\textwidth,
    height=0.50\textwidth,
    ybar,
    bar width=9pt,
    xmin=55.0,
    xmax=80.5,
    ymin=0,
    ymax=10500,
    xtick={55,60,65,70,75,80},
    minor x tick num=4,
    ytick={0,1000,...,10000},
    scaled y ticks=false,
    tick align=outside,
    tick style={black},
    xlabel={Codeword weight},
    ylabel={Number of codewords},
    title style={font=\bfseries},
    xlabel style={font=\small},
    ylabel style={font=\small},
    tick label style={font=\small},
    ymajorgrids,
    grid style={gray!25},
    axis line style={black},
    enlarge x limits=false,
]
\addplot[
    fill=red,
    draw=black,
] coordinates {
    (57,4608)
    (66,6912)
    (72,9216)
    (76,768)
    (78,9216)
    (79,3072)
};
\end{axis}
\end{tikzpicture}
    \caption{Weight distribution of low-weight codewords of $C({\mathcal E},384)$ obtained by Construction \ref{Cons_1}.}\label{fig:5g-weights}
\end{figure}
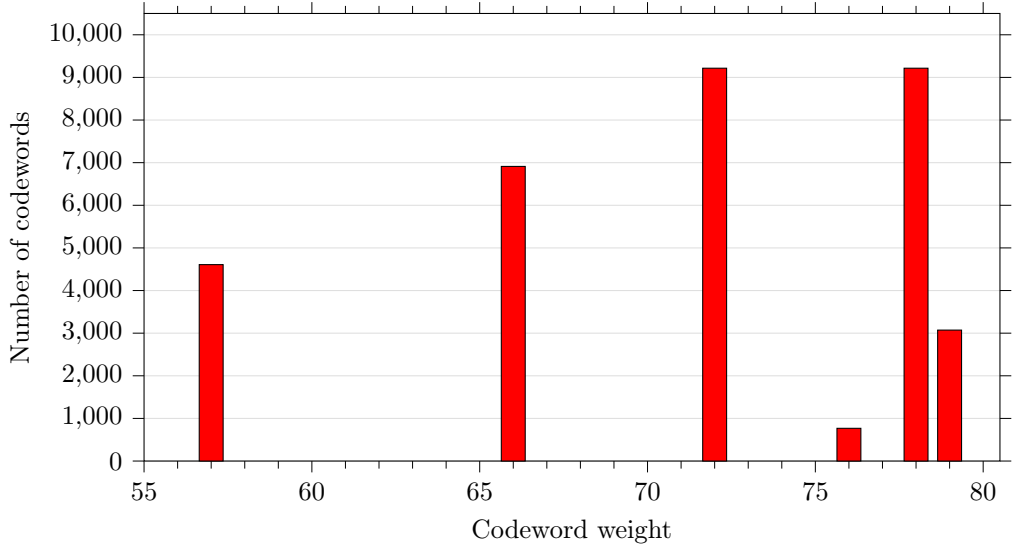

By puncturing the obtained small weight codewords of $C({\mathcal E},384)$ 
by $(46-j)384$ bits on the right, we obtain low-weight codewords of 
$C({\mathcal E}_{[j]},384)$. The upper bound results are summarized in Table~\ref{tab:vs-upper-bounds}. The corresponding low weight codewords and verification scripts are available in the MATLAB code \cite{M}. It turns out that the obtained low weight codewords are always codewords 
of the punctured codes $C'({\mathcal E}_{[j]},384)$, supposedly because of a structural
property of the code (the first two columns of the exponent matrix are  dense, Fig.~\ref{fig:bg1-bg2-mask}).

\begin{table}[h]
\centering
\caption{Upper bounds obtained from Vontobel-Smarandache 4-layers approach  on the minimum distances  the codes $C({\mathcal E}_{[j]},384)$. The same bounds also hold for left-punctured codes $C'({\mathcal E}_{[j]},384)$.}
\label{tab:vs-upper-bounds}
\begingroup
\renewcommand{\arraystretch}{1.12}
\setlength{\tabcolsep}{9pt}
\begin{tabular}{@{}c c @{\qquad} c c@{}}
\toprule
Layer range $j$ & Upper bound & Layer range $j$ & Upper bound \\
\midrule
$4$--$8$   & $14$ & $31$--$32$ & $36$ \\
$9$--$11$  & $18$ & $33$--$37$ & $40$ \\
$12$--$13$ & $22$ & $38$       & $44$ \\
$14$       & $24$ & $39$--$41$ & $47$ \\
$15$--$20$ & $26$ & $42$       & $50$ \\
$21$       & $29$ & $43$--$44$ & $54$ \\
$22$--$30$ & $32$ & $45$--$46$ & $57$ \\
\bottomrule
\end{tabular}
\endgroup
\end{table}

The idea behind Construction \ref{Cons_1} can be used to design  quasi-cyclic LDPC codes with a
parity-check matrix structure similar to that of the 5G codes~\cite{U}.   
In below we note that the approach is equivalent to choosing special 
subsets $S$ in Theorem~\ref{Th_VS}, however the realization via calculating extra parity $cH'^T$ is much more computationally saving than calculating the determinant directly.

\begin{proposition}
\label{prop:vs-extension-equivalence}
    Let $S'\subseteq\{1,\ldots,26\}$ be a $5$-subset, and let $c(S')$ be the
    codeword of $C({\mathcal E}_{[4]},384)$ obtained from Theorem~\ref{Th_VS}.
    Then
    \[
        (c(S'),c(S')H'^T)
    \]
    is the Vontobel--Smarandache codeword of $C({\mathcal E},384)$ associated with
    \[
        S'\cup\{27,\ldots,68\}.
    \]
\end{proposition}
\begin{proof}
    Directly from the fact that the code $C({\mathcal E}_{[4]},384)$ is an 
    iso-puncturing of the code $C({\mathcal E},384)$.
\end{proof}

\section{Minimum distances via reduced quasi-cyclic codes}\label{Sec_reduced}

The results represented in Subsection~\ref{sub41} are rather natural. 
We could find close but not exact statements in~\cite[Theorem~3]{ABA}, 
still, these are well-known results. For completeness, we provide their proofs 
in Appendix A.

\subsection{Basic results}
\label{sub41}

Let $u_i$, $i=1,\ldots,n$, be vectors of the same length $q$.
We denote by $Dub(u_1,\ldots,u_n,l)$ the concatenation of these vectors, 
where each vector is repeated $l$ times:
\[
    Dub(u_1,\ldots,u_n,l)=
    (\underbrace{u_1,\ldots,u_1}_{l\text{ times}},
     \underbrace{u_2,\ldots,u_2}_{l\text{ times}},
     \ldots,
     \underbrace{u_n,\ldots,u_n}_{l\text{ times}}).
\]

\begin{lemma}
\label{Dub_lemma}
    Let $E$ be a matrix with $n$ columns whose elements are integers  
    greater than or equal to $-1$, let $q$ be a divisor of $\overline{q}$, and let 
    $u_i,$ $i=1,\ldots n$ be vectors of length $q$. A vector $(u_1,\ldots,u_n)$
    is a codeword of $C(E,q)$ if and only if $Dub(u_1,\ldots,u_n,\overline{q}/q)$ 
    is a codeword of $C(E,\overline{q})$.
\end{lemma}

Let 
\[
    v = 
    \left(
       v_1^1,\ldots,v_1^{\overline{q}/q},
       v_2^1,\ldots,v_2^{\overline{q}/q},
       \ldots,
       v_n^1,\ldots,v_n^{\overline{q}/q}
    \right)
\]
be a vector of $C(E,\overline{q})$, where $v_i^j$, $i \in \{1,\ldots,n\}$, 
$j\in \{1,\ldots,\overline{q}/q\}$ are all of lengths $q$.
Define the mapping $R$ as follows:
\begin{equation}
\label{eq:r}
    R(v,q,\overline{q})=
    \left(
        \sum_{t=1}^{\overline{q}/q}v_1^t,
        \sum_{t=1}^{\overline{q}/q}v_2^t,
        \ldots,
        \sum_{t=1}^{\overline{q}/q}v_n^t
    \right).
\end{equation}

\begin{theorem}{\cite[Theorem~3]{ABA}}
\label{T1}
Let $E$ be a matrix  whose elements are integers  
    greater than or equal to $-1$ and let $q$ be a divisor of $\overline{q}$. 
    Then
    \begin{enumerate}[label=\arabic*.]
        \item $d(C(E,\overline{q})) \leq \frac{\overline{q}}{q}d(C(E,q))$.
        \item For a vector $v$ of $C(E,\overline{q})$ the vector $R(v, q, \overline{q})$
        defined by~\eqref{eq:r} is a codeword of $C(E,q)$. 
        \item Let $\overline{q}$ be $2^kq$. For a nonzero vector $v$ of $C(E,\overline{q})$
        the vector $R(v,q,\overline{q})$ defined by~\eqref{eq:r} is a nonzero codeword
        of $C(E,q)$ of weight not greater than that of $v$. In particular, 
        $d(C(E,q)) \leq d(C(E,\overline{q}))$. 
        \label{T1:item3}
    \end{enumerate}
\end{theorem}

\begin{corollary}
\label{C1} 
    Let $v$ be a nonzero codeword of $C(E,2^kq)$. Then the weight of $R(v,q,2^kq)$ is not 
    greater than that of $v$ and the block support of $R(v,q,2^kq)$ is a nonempty 
    subset of the block support of $v$.
\end{corollary}
\begin{proof} 
    By the way the mapping $R$ is defined, the block support of $R(v,q,2^kq)$ 
    is a subset of that of $v$. Moreover, due to Theorem~\ref{T1}.\ref{T1:item3}, it is nonempty.
\end{proof}

We finish by noting that the same statements hold for the punctured 5G codes
$C'({\mathcal E},q)$. They follow from the fact that the iso-punctured code
$C'({\mathcal E},q)$ has the same dimension as $C({\mathcal E},q)$ and all codewords of
$C({\mathcal E},q)$ are obtained from those of $C'({\mathcal E},q)$ by adding unique
$2q$ symbols.

\begin{theorem}
\label{T1_punc}
    Let $\mathcal{E}$ be an exponent matrix of a 5G LDPC code, and let $q$ be a divisor 
    of $\overline{q}$. Then
    \begin{enumerate}[label=\arabic*.]
        \item $d(C'(\mathcal{E},\overline{q})) \leq \frac{\overline{q}}{q}d(C'(\mathcal{E},q)) $.
        \item For a vector $v'$ of $C'({\mathcal E},\overline{q})$, the vector 
        $R(v',q,\overline{q})$ is a codeword of $C'(\mathcal{E},q)$. 
        \item Let $\overline{q}=2^kq$. For a nonzero vector $v'$ of 
        $C'({\mathcal E}, \overline{q})$, the vector $R(v',q,\overline{q})$ defined 
        by~\eqref{eq:r} is a nonzero codeword of $C'({\mathcal E},q)$ of weight not 
        greater than that of $v'$. In particular, 
        $d(C'({\mathcal E},q)) \leq d(C'({\mathcal E},\overline{q}))$.
    \end{enumerate}
\end{theorem}

\begin{corollary}
\label{C1_punc} 
    Let $\mathcal{E}$ be an exponent matrix of a 5G LDPC code and let $v'$ be a 
 nonzero    codeword of $C'({\mathcal E},2^kq)$. Then the weight of $R(v',q,2^kq)$ is not 
    greater than that of $v'$ and the block support of $R(v',q,2^kq)$ is a nonempty 
    subset of the block support of $v'$.
\end{corollary}
As a preambule for the application of these theoretical results we continue with discussion of the code from Example  \ref{ex:vs-polynomial-minors-ext}.
\begin{example} Consider the following polynomial shift matrix   
$$P=\left(\begin{array}{cccccccccccc}
   0& x^2& x^3& x^3& x^3& x& 1& 0& 0& 0& 0& 0 \\
     x^3& 0& x^3& 0& 1& 1& 1& 1& 0& 0& 0& 0 \\
     x& x& x^3& x& 0& 0& 0& 1& 1& 0& 0& 0 \\
     1& 1& 0& 1& 1& x& 0& 0& 1& 0& 0& 0 \\
     0& 0& 0& 0& 0& 0& 0& 0& 0& 1& 0& 0 \\
     0& 0& 0& 0& x^3& x^3& 0& 0& 0& 0& 1& 0 \\
     x^3& x& 0& x^3& 0& 0& 0& 0& 0& 0& 0& 1  
\end{array}\right),$$ denote  the respective exponent matrix  by $\overline{E}$. Consider the code $C(\overline{E},2)$, 
whose polynomial shift matrix is obtained from $P$ by replacing each $x^i$ with $x^{i \bmod 2}$. 
The exponent matrix of $C(\overline{E},2)$ is exactly the matrix $E$ from Example~\ref{ex:vs-polynomial-minors-ext}. The code $C(\overline{E},2)$ has minimum distance $3$ and contains only two weight three codewords, 
forming the quasi-cyclic orbit of \[
\bigl(x,0,0,x,0,0,1,0,0,0,0,0\bigr),
\]with block support $\{1,4,7\}$. 
By Corollary~\ref{C1}, any weight three codeword of $C(\overline{E},4)$ also has block support $\{1,4,7\}$, 
which greatly simplifies verifying that the minimum distance of $C(\overline{E},4)$ is $3$.

We write the codeword $v$ as $(a(x),b(x),c(x))$, where $a(x)$, $b(x)$, $c(x)$ reside in blocks $1$, $4$, $7$. By looking at the  equations, corresponding to the columns $1,4,7$ of $P$, and removing the all-zero checks we obtain the following system:$x^3b(x)+c(x)=0$,
$x^3a(x)+c(x)=0$,
$xa(x)+xb(x)=0$,
$a(x)+b(x)=0$,
$x^3a(x)+x^3b(x)=0.$
This leads to $a(x)=b(x)=1,c(x)=x^3$ up to quasi-cyclic shift and  $$(1,0,0,0,0,0,0,0,0,0,0,0,1,0,0,0,0,0,0,0,0,0,0,0,0,$$
$$0,0,1,0,0,0,0,0,0,0,0,0,0,0,0,0,0,0,0,0,0,0,0),
$$ is a codeword of $C(\overline{E},4)$, so the minimum distance of $C(\overline{E},4)$ is $3$.
\end{example}

\subsection{Reduced modular Brouwer--Zimmermann algorithms}

We recall the Brouwer--Zimmermann algorithm for finding the minimum 
distance~\cite{Grassl}. Without loss of generality, let $G$ 
be a generator matrix of a linear code $C$ of dimension $k$ and length $n$ in the canonical form, i.e.,
\[
    G=[I\mid A].
\]

\begin{algorithm}[H]
\caption{The Brouwer--Zimmermann algorithm for finding the minimum distance}
\label{alg:bz}
    \begin{algorithmic}[1]
        \Require generator matrix $G$.
        \State $d_U\gets n-k+1$ 
        \State $d_L\gets 1$
        \State $w\gets 1$
        \While{$w\leq k$ \textbf{and} $d_L<d_U$}
            \State $d_U\gets \min \left( d_U,\, w+\min \{ \operatorname{wt}(uA) 
            \mid u \in F^k,\ \operatorname{wt}(u)=w \} \right)$
            \State $d_L\gets w+1$
            \State $w\gets w+1$
        \EndWhile
        \State \Return $d_U$
    \end{algorithmic}
\end{algorithm}
To verify that the code distance is at least $d_L$, we apply a variant of the 
Brouwer--Zimmermann algorithm using Magma's 
\texttt{Verify\-Minimum\-Distance\-Lower\-Bound} intrinsic with input $d_L$.

Both of the above algorithms support a ``parallel'' realization, where several 
generator matrices with pairwise disjoint information sets are constructed, 
providing a faster growth of the lower bound $d_L$ than nonparallel versions~\cite[Algorithm 2.4]{Grassl}. 

\begin{algorithm}[H]
\caption{Parallel Brouwer--Zimmermann algorithm for finding the minimum distance}
\label{alg:parallel-bz}
    \begin{algorithmic}[1]
        \Require generator matrices $G_1,\ldots,G_r$ with disjoint information sets.
        \State $d_U\gets n-k+1$
        \State $d_L\gets 1$
        \State $w\gets 1$
        \While{$w\leq k$ \textbf{and} $d_L<d_U$}
            \State $d_U\gets \min \left( d_U,\, \min \{ \operatorname{wt}(uG_i) 
            \mid u \in F^k,\ \operatorname{wt}(u)=w,\ i=1,\ldots,r \} \right)$
            \State $d_L\gets (w+1)r$
            \State $w\gets w+1$
        \EndWhile
        \State \Return $d_U$
    \end{algorithmic}
\end{algorithm}

Straightforwardly, disjoint information sets exist only for low-rate codes, as 
the dimension is necessarily at most half the length of the code. There are 
variations of the parallel Brouwer-Zimmermann algorithm for high-rate codes 
using the notion of a partial information set~\cite[Algorithm 2.5]{Grassl}, 
\cite{Bouyuklieva}, but the lower bound growth for such codes is not as rapid 
as for low-rate codes.

\begin{approach}[Minimum distance lower bound for 5G LDPC codes]
\label{app:mdlb}
    A general idea can be described as follows. Let $E$ be an exponent matrix.  
    For any quasi-cyclic code $C(E,\overline{q})$, $\overline{q}=2^kq$, in 
    view of Theorem 1.3 a lower bound via modular reduction can be obtained as follows:
    \[
        d_L\leq d(C(E,\overline{q})),
    \]
    where $d_L$ is $d(C(E,q))$ and calculated via Brouwer--Zimmermann 
    algorithm/some other method (or $d_L$ is a lower bound on $d(C(E,q))$). 
\end{approach}

Now consider the estimation of the minimum distance for the all-layer 5G LDPC codes 
$C({\mathcal E},384)$ and $C({\mathcal E}_{[4]},384)$ with the $46\times68$ BG1 exponent 
matrix $\mathcal{E}$ and its submatrix ${\mathcal E}_{[4]}$. The minimum distances of 
$C({\mathcal E},q)$ for $q=3$ and $q=6$ are $8$ and $14$, respectively, directly applying Magma's function  \texttt{MinimumDistance}. For the code with 
$q=12$, using \texttt{VerifyMinimumDistanceLowerBound}, we can calculate
\[
    d\bigl(C({\mathcal E},12)\bigr)\geq 22.
\]
By Theorem~\ref{T1}.\ref{T1:item3}, $22$ is a lower bound 
on the minimum distance of the codes with $\overline{q}=12\cdot 2^i$, $i\geq 0$, including 
$\overline{q}=384$:
\[
    d\bigl(C({\mathcal E},2^k\cdot3)\bigr)\geq 22,\quad k\geq 2.
\]

On the other hand, a direct application of \texttt{VerifyMinimumDistanceLowerBound} 
for the code with $\overline{q}=384$ is not available in Magma due to limitations on
the code length.

Based on Theorem~\ref{T1_punc}, the same concept applies to punctured codes, and 
\[
    d\bigl(C'({\mathcal E},384)\bigr)\geq d\bigl(C'({\mathcal E},12)\bigr)\geq 22.
\]
 
Similarly, we evaluate the minimum distance of the 24-layer code $C({\mathcal E}_{[24]},384)$: 
\[
    d\bigl(C({\mathcal E}_{[24]},384)\bigr),\ d\bigl(C'({\mathcal E}_{[24]},384)\bigr)\geq 
    d\bigl(C({\mathcal E}_{[24]},12)\bigr) = d\bigl(C'({\mathcal E}_{[24]},12)\bigr) = 13,
\]
where the minimum distances of the codes $C({\mathcal E}_{[24]},12)$ and 
$C'({\mathcal E}_{[24]},12)$ were found by the built-in parallel Brouwer--Zimmermann
algorithm in Magma.

\begin{approach}[Minimum distance lower bound for high-rate 5G LDPC codes]
\label{app:mdlb_hr5g}
    For high-rate 5G LDPC codes, for example the $4$- and $6$-layer codes, we use a different 
    approach. The minimum distances of these codes are rather small and grow very slowly 
    with an increase in the circulant size. For circulant sizes $3$, $6$, $12$, and $24$ in 
    BG1, a partial weight distribution can be computed directly using Magma's function 
    \texttt{Words}$(\cdot)$; see the results in Tables~\ref{Table_1} and~\ref{Table_p}.
\end{approach}

\begin{table}[h]
\caption{Low-weight spectra distribution of the 5G $4$-layer BG1 code, calculated with Magma intrinsic \texttt{Words}$(\cdot)$.}
\label{Table_1}
\begin{tabular}{|l|r|r|r|r|r|r|r|r|}
    \hline
    Code & \multicolumn{8}{c|}{Weight} \\
    \cline{2-9}
     & 1 & 2 & 3 & 4 & 5 & 6 & 7 & 8 \\
    \hline
    $C({\mathcal E}_{[4]},3)$ & 0 & 9 & 6 & 297 & 1\,317 & 21\,895 & 102\,864 & 1\,181\,970 \\
    \hline
    $C({\mathcal E}_{[4]},6)$ & 0 & 6 & 0 & 87 & 288 & 5\,122 & ? & ? \\
    \hline
    $C({\mathcal E}_{[4]},12)$ & 0 & 0 & 0 & 90 & 372 & ? & ? & ? \\
    \hline
    $C({\mathcal E}_{[4]},24)$ & 0 & 0 & 0 & 0 & $96^*$ & ? & ? & ? \\
    \hline
\end{tabular}

\medskip
\footnotesize
$^*$ All $96$ words have the following block supports: 
$\{4,9,21,22,24\}$, $\{4,5,8,16,26\}$, $\{18,20,24,25,26\}$, $\{5,8,12,19,25\}$.
\end{table}

\begin{table}[h]
\caption{Low-weight spectrum of the left-punctured 5G $4$-layer BG1 code.}
\label{Table_p}
\begin{tabular}{|l|r|r|r|r|r|r|}
\hline
Code & \multicolumn{6}{c|}{Weight} \\
\cline{2-7}
 & 1 & 2 & 3 & 4 & 5 & 6  \\
\hline
$C'({\mathcal E}_{[4]},3)$ & 0&45 & 968 & 15\,852 & 218\,232 & 2\,442\,557 \\
\hline
$C'({\mathcal E}_{[4]},6)$ & 0 & 12 & 208 & 6\,621 & 140\,460 & 2\,979\,850 \\
\hline
$C'({\mathcal E}_{[4]},12)$ & 0 & 0 & $28^{*}$ & $540^{*}$ & $7\,980^{*}$ & $180\,720^{*}$ \\
\hline
$C'({\mathcal E}_{[4]},24)$ & 0 & 0 & 0 & 0 & $240^{**}$ & $?$  \\
\hline
\end{tabular}

\medskip
\footnotesize
\begin{minipage}{0.94\linewidth}
    $^*$ The words have $13\,485$ block supports\\
    $^{**}$ The codewords of $C'({\mathcal E},24)$ of weight five have only $4$ block supports
\end{minipage}
\end{table}







\setcounter{algorithm}{2}
\begin{algorithm}[h]
\caption{High-level description of the reduced circulant algorithm for 
a lower bound on the minimum distance of quasi-cyclic LDPC codes}
\label{alg:reduced-certification}
    \begin{algorithmic}[1]
        \Require $(q,j,E_{r\times n},B,d_L)$,
        \Statex \hspace{\algorithmicindent} where $q$ is divisible by $2^j$, $B$ 
        is the collection of block supports of all codewords of $C(E,q/2^j)$ 
        of weights at most $d_L$; 
        \Statex \hspace{\algorithmicindent} alternatively, $B$ consists of all $d_L$-subsets of $\{1,\ldots,n\}$
        \ForAll{$S\in B$}
            \State check whether $b=d(C(E^S,q))\geq d_L+1$
            \Statex \Comment{For a punctured code, check $b=d(C'(E^S,q))\geq d_L+1$}
            \If{$b=0$}
                \State \textbf{break}
            \EndIf
        \EndFor
        \If{$b=1$}
            \State \Return $d(C(E,q))\geq d_L+1$
        \Else
            \State \Return $d(C(E,q))\leq d_L$
        \EndIf
        \Statex \Comment{For a punctured code, replace $C(E,q)$ by $C'(E,q)$.}
    \end{algorithmic}
\end{algorithm}

\begin{algorithm}[h]
\caption{High-level description of the reduced circulant algorithm for finding 
the block supports for small-weight words of quasi-cyclic LDPC codes}
\label{alg:reduced-filtering}
    \begin{algorithmic}[1]
        \Require the same input as in Algorithm~\ref{alg:reduced-certification}
        \State $B'\gets\emptyset$
        \ForAll{$S\in B$}
            \State check whether $b=d(C(E^S,q))\geq d_L+1$
            \If{$b=0$}
                \State $B'\gets B'\cup\{S\}$
            \EndIf
        \EndFor
        \State \Return the collection $B'$ of block supports of nonzero 
        words of $C(E,q)$ with weight at most $d_L$
    \end{algorithmic}
\end{algorithm}

For $q=48$, solving the problem via the Brouwer--Zimmermann algorithm in Magma 
directly required a substantial amount of computational time ($472\,797$ seconds 
on a personal computer to obtain a lower bound of $6$), and it becomes absolutely 
computationally infeasible for $q=96$ and $q=192$, despite the fact that it can be 
solved in a couple of hours using the approach with a high-level description 
provided in Algorithms~3 and~4 and the details below.

We gradually increase the circulant sizes starting from $48$: $48$, $96$, and $192$. 

\textbf{Case $q=48$, $96$.}
The minimum distance of $C({\mathcal E}_{[4]},24)$ is $5$, and there are exactly 
$96$ codewords of weight $5$ having four different block supports:
\begin{equation}
\label{e_supp}
    \begin{gathered}
        \{4,9,21,22,24\},\quad
        \{4,5,8,16,26\},\\
        \{18,20,24,25,26\},\quad
        \{5,8,12,19,25\}.
    \end{gathered}
\end{equation}
By Theorem~\ref{T1}.\ref{T1:item3}, we observe that the minimum distance of $C({\mathcal E}_{[4]},48)$ 
is also at least $5$, and the reduction mapping $R(v,24,48)$~\eqref{eq:r} 
sends any weight $5$ codeword $v$ of $C({\mathcal E}_{[4]},48)$ to a weight $5$ codeword of $C({\mathcal E}_{[4]},24)$, 
because the minimum distance of the latter is $5$. Also, by Corollary~\ref{C1}, 
the block support of $R(v,24,48)$ is one of the sets~(\ref{e_supp}) for the weight $5$
codewords of $C({\mathcal E}_{[4]},24)$. Therefore, checking whether $d\bigl(C({\mathcal E}_{[4]},48)\bigr)$ 
is $5$ boils down to checking whether $d\bigl(C({\mathcal E}_{[4]}^S,48)\bigr)$ is $5$ 
for those four block supports $S$ from~(\ref{e_supp}) such that $d\bigl(C({\mathcal E}_{[4]}^S,24)\bigr)=5$. 
This solution also benefits from the fact that the code $C({\mathcal E}_{[4]}^S,48)$ 
has an exponent matrix ${\mathcal E}_{[4]}^S$ with only $4$ rows and $5$ columns.
Due to the rate of the code $C({\mathcal E}_{[4]}^S,48)$ being $\frac{1}{5}$, most 
of the time the Brouwer--Zimmermann algorithm uses its parallel version when checking 
whether the lower bound on $d\bigl(C({\mathcal E}_{[4]}^S,24)\bigr)$ is $6$. These jobs were finished 
with a positive answer, indicating that 
\[
    d\bigl(C({\mathcal E}_{[4]},48)\bigr)\geq 6.
\]
The idea is generalized in Algorithm 3, which we use as follows: 
\begin{enumerate}[label=\arabic*.,leftmargin=2em,itemsep=0.25ex,topsep=0.5ex]
    \item The weight $5$ codewords of $C({\mathcal E}_{[4]},24)$ are found using \texttt{Words}$(\cdot)$, and $B$ is the set of their block supports.
    \item Algorithm~\ref{alg:reduced-certification} with $q=48$, $j=1$, $E={\mathcal E}_{[4]}$, $B$, and $d_L=5$ gives
          \[
              d(C({\mathcal E}_{[4]},48))\geq d_L+1=6.
          \]
\end{enumerate}

On the other hand, studying the $6$-subsets of $\{1,\ldots,26\}$, we found a set
\[
    S=\{4,7,13,14,17,23\}
\]
with $d(C({\mathcal E}_{[4]}^S,96))=6$, so from Theorem \ref{T1}.\ref{T1:item3} we have
\[
    6 \leq d(C({\mathcal E}_{[4]},48))\leq d(C({\mathcal E}_{[4]},96))=6
\]
and therefore
 $d(C({\mathcal E}_{[4]},48))=d(C({\mathcal E}_{[4]},96))=6$.

{\bf Case $q=192$}. 
From the above we have $ d(C({\mathcal E}_{[4]},96))=6,$ which is a lower bound 
on $d(C({\mathcal E}_{[4]},192))$ by Theorem \ref{T1}.\ref{T1:item3} We improve the lower bound 
from $6$ to $8$ as follows. Any weight $6$ or $7$ codeword $v$ of $C({\mathcal E}_{[4]},192)$ 
has a block support of size $7$ or less. By Corollary \ref{C1} the block support contains  
the block support of the reduced nonzero word $R(v,48,192)$ of the code $C({\mathcal E},48)$.

We study all $7$-subsets of $\{1,\ldots,26\}$. For any such $7$-subset $S$ we launch 
Brouwer--Zimmermann algorithm for the code $C({\mathcal E}_{[4]}^S,48)$ to check whether 
$d(C({\mathcal E}_{[4]}^S,48))$ (note that the circulant size is $48$ here) is at least $8$ 
(each of which are finished relatively fast).  The answer is negative only for $276$ subsets 
out of $(^{26}_7)=657800$. For each of these subsets, we launch Brouwer--Zimmermann algorithm 
to check whether $d(C({\mathcal E}_{[4]}^S,192))$ is at least $8$. These jobs are finished with 
positive answer, indicating that 
\[
    d(C({\mathcal E}_{[4]},192))\geq 8.
\]
The idea is generalized in Algorithm 4, which we use as follows: 
\begin{enumerate}[label=\arabic*.,leftmargin=2em,itemsep=0.25ex,topsep=0.5ex]
    \item Let $B$ be the set of all $7$-subsets of $\{1,\ldots,26\}$.
    \item Apply Algorithm~\ref{alg:reduced-filtering} with $q=48$, $j=1$, $E={\mathcal E}_{[4]}$, 
    $B$, and $d_L=7$. The output $B'$ contains the  block supports of $C({\mathcal E}_{[4]},48)$ (and $C({\mathcal E}_{[4]},192)$ as well) for weight $6$ and $7$ words.
    \item Apply Algorithm~\ref{alg:reduced-filtering} with $q=192$, $j=2$, $E={\mathcal E}_{[4]}$, 
    $B=B'$, and $d_L=7$. The output is an empty set, so
    \[
      d(C({\mathcal E}_{[4]},192))\geq 8.
    \]
\end{enumerate}

On the other hand, studying the $7$-subsets of $\{1,\ldots,26\}$, we found a set $S=\{5,8,9,15,18,25,26\}$ 
with $d(C({\mathcal E}_{[4]}^S,192))=8$, thus
\[
    d(C({\mathcal E}_{[4]},192))=8.
\]

{\bf Punctured codes.} 
Because of Theorem \ref{T1_punc} similar ideas work for the punctured codes $C'({\mathcal E}_{[4]},q)$. 
The low weight distribution for the codes $C'({\mathcal E}_{[4]},q)$, $q=3$, $6$, $12$, $24$ was 
obtained (see Table \ref{Table_p}) using the Magma intrinsic  \texttt{Words}$()$, and the code 
$C'({\mathcal E}_{[4]},24)$ has the same minimum distance $5$ as the nonpunctured code $C({\mathcal E}_{[4]},24)$.

All weight $5$ codewords of $C({\mathcal E}_{[4]},24)$ have $4$ block supports. Testing this collection $B$ in Algorithm 3 yields 
$6\leq d(C'({\mathcal E}_{[4]},48))$, which, combined with the fact that the nonpunctured codes 
$C({\mathcal E}_{[4]},48)$, $C({\mathcal E}_{[4]},96)$ have minimum distance $6$, gives us from (\ref{punc_min_dist}): 
$6\leq d(C'({\mathcal E}_{[4]},48))\leq d(C'({\mathcal E}_{[4]},96))\leq d(C({\mathcal E}_{[4]},96))\leq 6$, so
\[
    d(C'({\mathcal E}_{[4]},48))=d(C'({\mathcal E}_{[4]},96))=6.
\]

To establish lower bound 
\[
    7\leq d(C'({\mathcal E}_{[4]},192)),
\]
it is sufficient to exclude the case $d(C'({\mathcal E}_{[4]},192))=6$.  
The approach relies on Algorithm 4. The words of weights at most $6$ of the code $C'({\mathcal E}_{[4]},12)$ 
were found by the Magma function \texttt{Words}$()$ and have $13485$ block supports. By checking that 
$d(C'({\mathcal E}^S,24))\geq 7$, we exclude some of these block supports, and there are $4068$ 
block supports of the words of weight at most $6$ in the code $C'({\mathcal E}^S,24)$. Further, 
by checking that $d(C'({\mathcal E}^S,48))\geq 7$, we exclude most of the block supports, and 
there are just $22$ block supports of the words of weight at most $6$ in the code $C'({\mathcal E}^S,48)$. 
By checking that $d(C'({\mathcal E}^S,96))\geq 7$, we have just $1$ block support of the words 
of weight at most $6$ in the code $C'({\mathcal E}^S,96)$. Finally, we checked that $d(C'({\mathcal E}^S,192))\geq 7$ 
for the only such subset $S$, indicating that $d(C'({\mathcal E},192))\geq 7$. Using Algorithm 4, 
we followed the steps below. 
\begin{enumerate}[label=\arabic*.,leftmargin=2em,itemsep=0.25ex,topsep=0.5ex]
    \item Let $B$ be the set of block supports of codewords of $C'({\mathcal E}_{[4]},12)$ of weight at most $6$; these supports are found using \texttt{Words}$()$, and $|B|=13485$.
    \item Apply Algorithm~\ref{alg:reduced-filtering} with $j=1$, $q=24$, $d_L=6$, and $B$; the output $B'$: $|B|=4068$.
    \item Apply Algorithm~\ref{alg:reduced-filtering} with $j=1$, $q=48$, $d_L=6$, and $B=B'$; the output $B'$: $|B|=22$.
    \item Apply Algorithm~\ref{alg:reduced-filtering} with $j=1$, $q=96$, $d_L=6$, and $B=B'$; the output $B'$: $|B|=1$.
    \item Apply Algorithm~\ref{alg:reduced-filtering} with $j=1$, $q=192$, $d_L=6$, and $B=B'$; the output $B'$: $|B|=0$, which indicates that
          \[
              d(C'({\mathcal E}_{[4]},192))\geq 7.
          \]
\end{enumerate}
\begin{remark}
    The results of the minimum distance study for 5G LDPC codes are summarized in Table \ref{T3}. 
    We see that the minimum distances of the punctured and nonpunctured codes are either similar 
    or within the same lower and upper bounds.
\end{remark}

Mostly, the minimum distances of the high-rate 5G LDPC codes are small, which, in particular, 
probably results in a noticeable LDPC UIBLER in their decoding compared to a random code; 
see Fig.~\ref{fig:LDPC_UIBLER_intro}. On the pro side, the LDPC UIBLER is easily caught by 
verifying the CRC block. Moreover, the 5G code design arguably suits better the HARQ protocol 
and low-complexity encoding in hardware. We proceed further with exploiting modulo reduction 
ideas for early termination of decoding5G codes in the next section. 
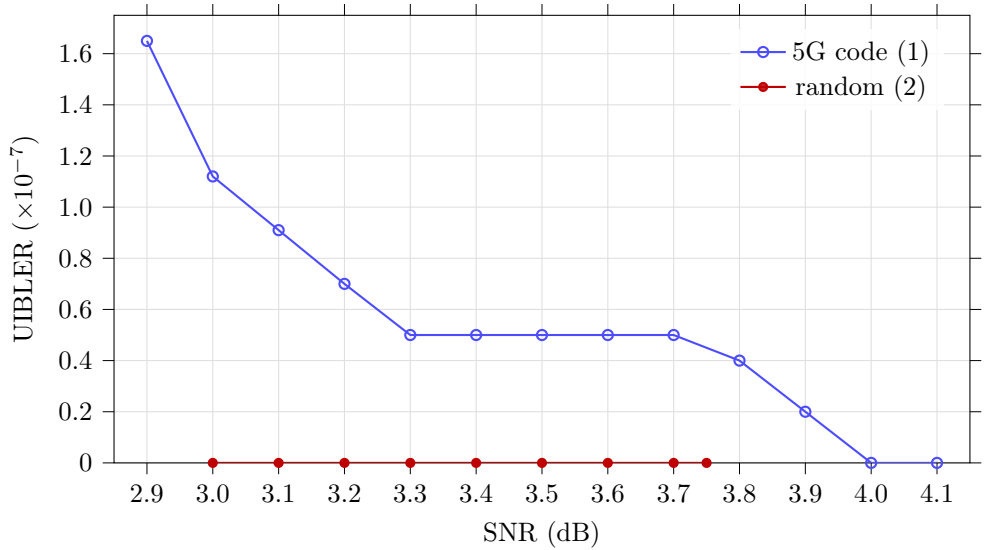
\begin{figure}[H]
    \begin{tikzpicture}
\begin{axis}[
    width=0.86\textwidth,
    height=0.50\textwidth,
    xmin=2.85,
    xmax=4.15,
    ymin=0,
    ymax=1.75e-7,
    xtick={2.9,3.0,3.1,3.2,3.3,3.4,3.5,3.6,3.7,3.8,3.9,4.0,4.1},
    xticklabels={$2.9$,$3.0$,$3.1$,$3.2$,$3.3$,$3.4$,$3.5$,$3.6$,$3.7$,$3.8$,$3.9$,$4.0$,$4.1$},
    ytick={0,2e-8,4e-8,6e-8,8e-8,1e-7,1.2e-7,1.4e-7,1.6e-7},
    yticklabels={$0$,$0.2$,$0.4$,$0.6$,$0.8$,$1.0$,$1.2$,$1.4$,$1.6$},
    scaled y ticks=false,
    tick align=outside,
    tick style={black},
    xlabel={SNR (dB)},
    ylabel={UIBLER ($\times 10^{-7}$)},
    xlabel style={font=\small},
    ylabel style={font=\small},
    tick label style={font=\small},
    legend style={
        font=\small,
        at={(0.97,0.97)},
        anchor=north east,
        draw=none,
        fill=white,
        fill opacity=0.85,
        text opacity=1
    },
    ymajorgrids,
    xmajorgrids,
    grid style={gray!25},
    axis line style={black},
]
\addplot[
    thick,
    mark=o,
    mark size=1.9pt,
    color=blue!70,
] coordinates {
    (2.9,1.65e-7)
    (3.0,1.12e-7)
    (3.1,9.1e-8)
    (3.2,7.0e-8)
    (3.3,5.0e-8)
    (3.4,5.0e-8)
    (3.5,5.0e-8)
    (3.6,5.0e-8)
    (3.7,5.0e-8)
    (3.8,4.0e-8)
    (3.9,2.0e-8)
    (4.0,0)
    (4.1,0)
};
\addlegendentry{5G code (1)}

\addplot[
    thick,
    mark=*,
    mark size=1.5pt,
    color=red!75!black,
] coordinates {
    (3.0,0)
    (3.1,0)
    (3.2,0)
    (3.3,0)
    (3.4,0)
    (3.5,0)
    (3.6,0)
    (3.7,0)
    (3.75,0)
};
\addlegendentry{random (2)}
\end{axis}
\end{tikzpicture}
    \caption{LDPC UIBLER for 5G vs same rate and information block random code on a linear 
    scale with early termination by full  syndrome. Random code shows zero UIBLER arguebly 
    due to larger minimum distance.}\label{fig:LDPC_UIBLER_intro}
\end{figure}

\begin{table}[H]
\caption{Minimum-distances bounds for BG1 5G LDPC codes with $q\mid 384$ and $q\geq 48$. '--' denotes the same as above}
\label{T3}
\begin{tabular}{|l|c|c|c|c|}\hline
\,\,\,\,\,\,\,\,\,\,Code & Dim. & Code    & Lower  & Upper \\
     &      & length  & Bound  & Bound \\ \hline
$C({\mathcal E}_{[4]},48)$    & 1056 & 1248 & 6 (Approach~\ref{app:mdlb_hr5g}) & 6 (Approach~\ref{app:mdlb_hr5g}) \\
$C'({\mathcal E}_{[4]},48)$   &      & 1152 &                                  &                                  \\ \hline
$C({\mathcal E}_{[4]},96)$    & 2112 & 2496 & -- & -- \\
$C'({\mathcal E}_{[4]},96)$   &      & 2304 &                                  &                                  \\ \hline
$C({\mathcal E}_{[4]},192)$   & 4224 & 4992 & 8 (Approach~\ref{app:mdlb_hr5g}) & 8 (Approach~\ref{app:mdlb_hr5g}) \\
$C'({\mathcal E}_{[4]},192)$  &      & 4608 & 7 (Approach~\ref{app:mdlb_hr5g}) &                                  \\ \hline
$C({\mathcal E}_{[4]},384)$   & 8448 & 9984 & 8 (Theorem~\ref{T1}.3)      & 14 (Section~\ref{subsec:sec_VS_4lay}) \\
$C'({\mathcal E}_{[4]},384)$  &      & 9216 & 7 (Theorem~\ref{T1_punc}.3) &                                       \\ \hline
$C({\mathcal E}_{[i]},384)$   & -- & $(22+i)384$ & --      & -- \\
$C'({\mathcal E}_{[i]},384)$  &      & $(20+i)384$ & -- &                                       \\
$5\leq i\leq 8$               &      &             &                             &                                       \\ \hline
$C({\mathcal E}_{[i]},384)$   & -- & $(22+i)384$ & --      & 18 (Section~\ref{subsec:sec_VS_4lay}) \\
$C'({\mathcal E}_{[i]},384)$  &      & $(20+i)384$ & -- &                                       \\
$9\leq i\leq 11$              &      &             &                             &                                       \\ \hline
$C({\mathcal E}_{[i]},384)$   & -- & $(22+i)384$ & --      & 22 (Section~\ref{subsec:sec_VS_4lay}) \\
$C'({\mathcal E}_{[i]},384)$  &      & $(20+i)384$ & -- &                                       \\
$12\leq i\leq 13$             &      &             &                             &                                       \\ \hline
$C({\mathcal E}_{[14]},384)$  & -- & 13824 & --      & 24 (Section~\ref{subsec:sec_VS_4lay}) \\
$C'({\mathcal E}_{[14]},384)$ &      & 13056 & -- &                                       \\ \hline
$C({\mathcal E}_{[i]},384)$   & -- & $(22+i)384$ & --      & 26 (Section~\ref{subsec:sec_VS_4lay}) \\
$C'({\mathcal E}_{[i]},384)$  &      & $(20+i)384$ & -- &                                       \\
$15\leq i\leq 20$             &      &             &                             &                                       \\ \hline
$C({\mathcal E}_{[i]},384)$   & -- & $(22+i)384$ & --      & 32 (Section~\ref{subsec:sec_VS_4lay}) \\
$C'({\mathcal E}_{[i]},384)$  &      & $(20+i)384$ & -- &                                       \\
$21\leq i\leq 23$             &      &             &                             &                                       \\ \hline
$C({\mathcal E}_{[24]},384)$  & -- & 17664 & 13 (Approach~\ref{app:mdlb}) & -- \\
$C'({\mathcal E}_{[24]},384)$ &      & 16896 & 13 (Approach~\ref{app:mdlb}) &                                       \\ \hline
$C({\mathcal E},384)$         & -- & 26112 & 22 (Approach~\ref{app:mdlb}) & 57 (Section~\ref{subsec:sec_VS_4lay}) \\
$C'({\mathcal E},384)$        &      & 25344 & 22 (Approach~\ref{app:mdlb}) &                                       \\ \hline
\end{tabular}
\end{table}

\section{LDPC early termination based on reduced shifters}\label{Sec_ET}
A natural idea to reduce the complexity of syndrome calculation for 5G LDPC codes was proposed in \cite{I} and is based on the decomposition~\eqref{decomp}. Since all codes in the 5G standard involve more than four layers of parity checks~\eqref{eq:all_layer_codes}, it was suggested to terminate decoding early when the word passes the checks corresponding to the first four layers of the parity-check matrix. We refer to this method as the $4$-layer early termination approach for LDPC decoding.

Further, the ideas described in the previous section lead us to the following reduction of the syndrome calculation complexity.

Let $E = \begin{pmatrix} E'\\ E''\\ E''' \end{pmatrix}$ be the exponent matrix of a quasicyclic code $C(E,q)$, where $E'$, $E''$, and $E'''$ have $t'$, $t''$, and $t'''$ rows, respectively, and $t = t' + t'' + t'''$ is the total number of rows of $E$. A word $v$ is a codeword of $C(E,q)$ if and only if
\[
H'v^{T} = 0,\quad
H''v^{T} = 0,\quad
H'''v^{T} = 0,
\]
where $H'$, $H''$, and $H'''$ are obtained by unrolling $E'$, $E''$, and $E'''$, respectively, with circulants of size $q$.
By Theorem~\ref{T1}, if $v$ is a codeword of $C(E'',q)$, then for a divisor $q'$ of $q$, the vector $R(v,q',q)$, defined in~\eqref{eq:r}, is a codeword of the ``small code'' $C(E'',q')$ with parity-check matrix $\widetilde H''$, which has $q/q'$ times fewer rows and columns than $H''$. We say that the {\it early termination} condition $T(t'_{q},t''_{q'})$ holds for $v$ if
\[
H'v^{T} = 0,
\]
and
\[
\widetilde H'' R(v,q',q)^{T} = 0.
\]

Early termination conditions are used with the min-sum decoder, see Algorithm~\ref{alg:early-termination}. The results of this algorithm with various early termination techniques for the 5G code and a random code are shown in Figs.~\ref{FigET_IBLER} and~\ref{FigET_UIBLER}. 

We now outline the most important subcases.

\textbf{Full syndrome early termination }
$T(t_q,0)$. We declare LDPC decoding successful for a word $v$ if $H'v^{T} = 0$, $H''v^{T} = 0$, and $H'''v^{T} = 0$, that is, only when $v$ is a codeword of the quasi-cyclic code $C(E,q)$. This is the most reliable termination approach from the undetected error point of view. In Figs.~\ref{FigET_IBLER} and~\ref{FigET_UIBLER}, this case corresponds to the curves shown in light blue and red, namely $T(6_{384},0)$ and $T(8_{192},0)$. This approach was also used in the experiments shown in Figs.~\ref{fig:ibler_intro} and~\ref{fig:average_iterations_intro}.

\textbf{$4$-layer early termination}   $T(4_q,0)$.  In this case we completely disregard two last check layers even in the mod $q'$ reduced form. The option works only for 5G LDPC code and is arguebly related to matrix design and decomposition (\ref{decomp}) in particular. This idea  was proposed in \cite{I}.  

We suggest going further and reducing the last two layers of the parity-check matrix $H$ modulo a small factor $q'$ of $q = 384$ (e.g., $q' = 16$), and we propose the early termination criterion $T(2_{384},2_{16})$.

\textbf{Implementation note.}
All simulations  used the same C++ decoder core. It implements floating point layered normalized min-sum decoder, corresponding to the \texttt{norm-min-sum} mode of the MATLAB LDPC decoder \cite{MathWorksLDPCDecode}. The min-sum update rule itself is unchanged: the program computes check node signs, the smallest and second smallest incoming magnitudes, applies the normalization factor, and updates LLRs layer by layer. The implemented modification is only the early termination module. It selects which syndrome test is applied to the hard-decision word after an iteration.

More precisely, the C++ program stores the lifted parity-check matrix in a sparse row representation. For every row it keeps the position of the first nonzero entry in a global column-index array, the row weight, and the indices of all nonzero columns. This representation is used both by the layered min-sum update and by the syndrome calculation. Therefore, the experiments do not compare different decoders; they compare the same decoder with different values of the termination modes.

\begin{algorithm}[h]
\caption{Layered normalized min-sum decoding with configurable early termination}
\label{alg:early-termination}
    \begin{algorithmic}[1]
        \Require LLR vector $L$, exponent matrix $E$,
        maximum number of iterations $I_{\max}$, termination mode $T(t'_{q},t''_{q'})$
        \Ensure hard-decision word $v$ and termination flag
        \State Initialize variable LLRs and check-to-variable messages
        \For{$i=1,\ldots,I_{\max}$}
            
            Perform one iteration of the layered min-sum decoder.
            \State Find the hard-decision word $v$. Report successful decoding if $T(t'_{q},t''_{q'})$ holds for $v$.
            
        \EndFor
        \State \Return $v$, failure
    \end{algorithmic}
\end{algorithm}

Early termination is evaluated only after a complete iteration has been processed. The hard decision is not stored as a separate vector during the iterations; the syndrome routine reads it directly from the signs of the current LLRs. A negative LLR is interpreted as bit one, and a nonnegative LLR as bit zero. Thus, the full and partial syndrome checks are computed by XOR-summing the signs over the sparse row representation.

The reduced check is evaluated without constructing a new full parity-check matrix. In the C++ code, a reduced layer is handled by scanning the selected rows of the lifted matrix and accumulating their parities into $q$ bins according to the row index modulo $q$. This implements the operation $R(v,q,384)$ in the stopping test. For example, in the $T(2_{384},2_{16})$ early termination the first two layers give $2\cdot 384$ ordinary parity checks, while the next two layers give only $2\cdot 16$ reduced checks. 

The simulator records whether the decoder stopped before the maximum number of iterations and returns the iteration count. The same output word is then used to count information-block errors and undetected LDPC errors.

\textbf{Conclusion.} For 5G  decoder degradation under alternative termination criterias is insignificant, given that 5G is targeted at retransmission HARQ scheme. Actually, the IBLER plots for $T(4_{384},0)$ termination vs $T(6_{384},0)$ termination (full syndrome) 5G code are surprisingly close  with a  minor difference of $4*10^{-8}$ for error rate only at low SNR point at $2.9$, and the IBLER results for other points were identical (we ran $10^8$ frames for each point). 
This provides a solid ground for choosing the latter in \cite{I} and as a priority early termination option for OCT standard \cite{SDA}.

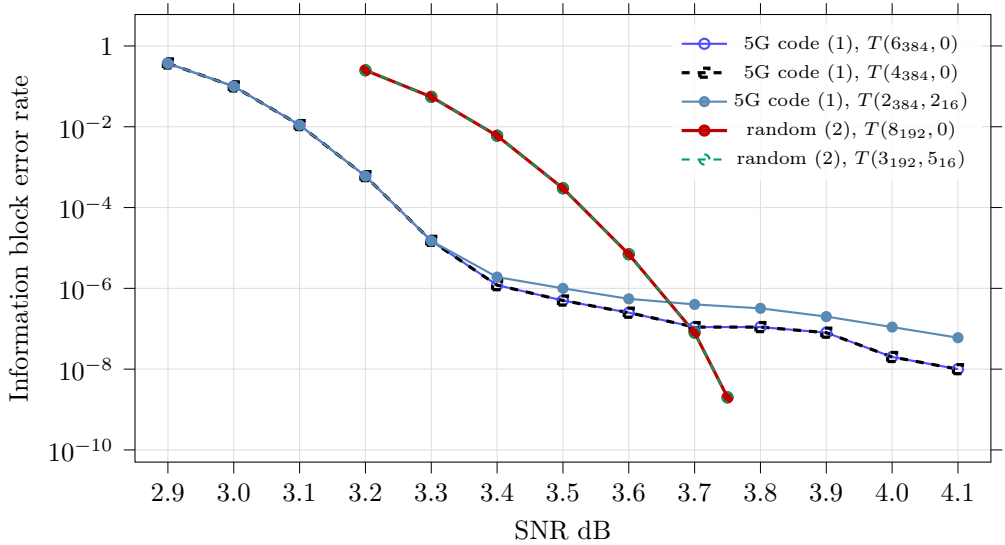
\begin{figure}[H]
\centering
\begin{tikzpicture}
\begin{semilogyaxis}[
    width=0.86\textwidth,
    height=0.50\textwidth,
    xmin=2.85,
    xmax=4.15,
    ymin=5e-11,
    ymax=6,
    xtick={2.9,3.0,3.1,3.2,3.3,3.4,3.5,3.6,3.7,3.8,3.9,4.0,4.1},
    xticklabels={$2.9$,$3.0$,$3.1$,$3.2$,$3.3$,$3.4$,$3.5$,$3.6$,$3.7$,$3.8$,$3.9$,$4.0$,$4.1$},
    ytick={1,1e-2,1e-4,1e-6,1e-8,1e-10},
    yticklabels={$1$,$10^{-2}$,$10^{-4}$,$10^{-6}$,$10^{-8}$,$10^{-10}$},
    minor y tick num=1,
    tick align=outside,
    tick style={black},
    xlabel={SNR dB},
    ylabel={Information block error rate},
    xlabel style={font=\small},
    ylabel style={font=\small},
    tick label style={font=\small},
    legend style={
        font=\tiny,
        at={(0.99,0.98)},
        anchor=north east,
        draw=none,
        fill=white,
        fill opacity=0.85,
        text opacity=1
    },
    ymajorgrids,
    xmajorgrids,
    grid style={gray!25},
    axis line style={black},
]
\addplot[
    thick,
    mark=o,
    mark size=1.9pt,
    color=blue!70,
] coordinates {
    (2.9,3.7000001e-1)
    (3.0,1.0e-1)
    (3.1,1.1e-2)
    (3.2,6.0e-4)
    (3.3,1.5e-5)
    (3.4,1.2e-6)
    (3.5,5.0e-7)
    (3.6,2.5e-7)
    (3.7,1.1e-7)
    (3.8,1.1e-7)
    (3.9,8.0e-8)
    (4.0,2.0e-8)
    (4.1,1.0e-8)
};
\addlegendentry{5G code (1), $T(6_{384},0)$}

\addplot[
    very thick,
    mark=square,
    dashed,
    mark size=1.7pt,
    color=black,
] coordinates {
    (2.9,3.7e-1)
    (3.0,1.0e-1)
    (3.1,1.1e-2)
    (3.2,6.0e-4)
    (3.3,1.5e-5)
    (3.4,1.2e-6)
    (3.5,5.0e-7)
    (3.6,2.5e-7)
    (3.7,1.1e-7)
    (3.8,1.1e-7)
    (3.9,8.0e-8)
    (4.0,2.0e-8)
    (4.1,1.0e-8)
};
\addlegendentry{5G code (1), $T(4_{384},0)$}

\addplot[
    thick,
    mark=*,
    mark size=1.7pt,
    color=cyan!55!violet,
] coordinates {
    (2.9,3.700001e-1)
    (3.0,1.000001e-1)
    (3.1,1.100001e-2)
    (3.2,6.000001e-4)
    (3.3,1.500001e-5)
    (3.4,1.9e-6)
    (3.5,1.0e-6)
    (3.6,5.5e-7)
    (3.7,4.0e-7)
    (3.8,3.2e-7)
    (3.9,2.0e-7)
    (4.0,1.1e-7)
    (4.1,6.0e-8)
};
\addlegendentry{5G code (1), $T(2_{384},2_{16})$}

\addplot[
    very thick,
    mark=*,
    mark size=1.7pt,
    color=red!75!black,
] coordinates {
    (3.2,2.5e-1)
    (3.3,5.5e-2)
    (3.4,6.0e-3)
    (3.5,3.0e-4)
    (3.6,7.0e-6)
    (3.7,8.0e-8)
    (3.75,2.0e-9)
};
\addlegendentry{random (2), $T(8_{192},0)$  }

\addplot[
    thick,
    dashed,
    mark=o,
    mark size=1.9pt,
    color=teal!75!green,
] coordinates {
    (3.2,2.5e-1)
    (3.3,5.5e-2)
    (3.4,6.0e-3)
    (3.5,3.0e-4)
    (3.6,7.0e-6)
    (3.7,8.0e-8)
    (3.75,2.0e-9)
};
\addlegendentry{random (2), $T(3_{192},5_{16})$}

\end{semilogyaxis}
\end{tikzpicture}
\caption{IBLER for 5G LDPC code and random code with various early termination criterias}
\label{FigET_IBLER}
\end{figure}

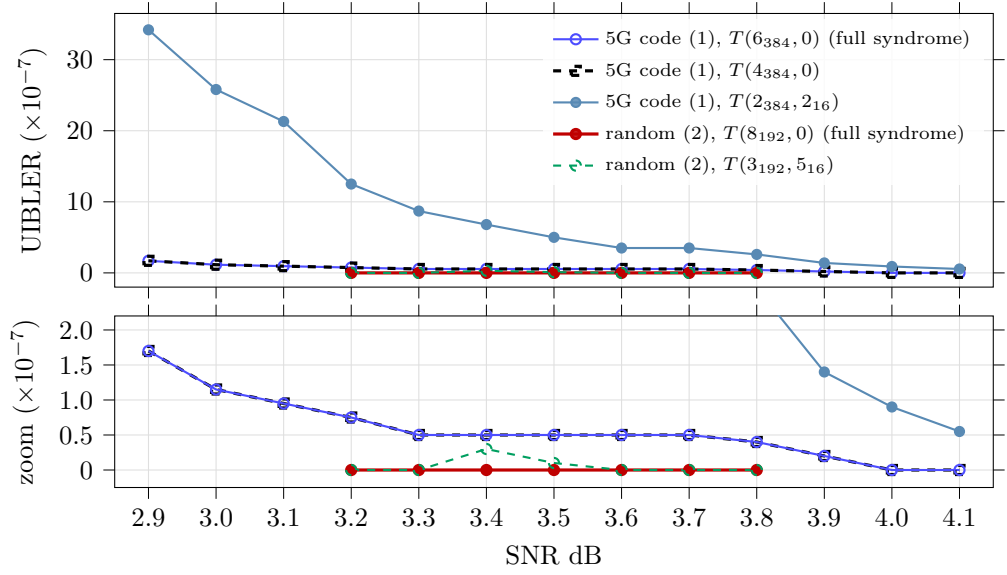
\begin{figure}[h]
\centering
\begin{tikzpicture}
\begin{axis}[
    at={(0,2.65cm)},
    anchor=south west,
    width=0.88\textwidth,
    height=5.2cm,
    xmin=2.85,
    xmax=4.15,
    ymin=-0.20e-6,
    ymax=3.65e-6,
    xtick={2.9,3.0,3.1,3.2,3.3,3.4,3.5,3.6,3.7,3.8,3.9,4.0,4.1},
    xticklabels=\empty,
    ytick={0,10.0e-7,20.0e-7,30.0e-7},
    yticklabels={$0$,$10$,$20$,$30$},
    scaled y ticks=false,
    tick align=outside,
    tick style={black},
    ylabel={UIBLER ($\times 10^{-7}$)},
    ylabel style={font=\small},
    tick label style={font=\small},
    legend style={
        font=\tiny,
        at={(0.99,0.98)},
        anchor=north east,
        draw=none,
        fill=white,
        fill opacity=0.85,
        text opacity=1,
        row sep=1pt
    },
    legend cell align=left,
    ymajorgrids,
    xmajorgrids,
    grid style={gray!25},
    axis line style={black},
]

    \addplot[
  thick,
    mark=o,
    mark size=1.9pt,
    color=blue!70
] coordinates {
    (2.9,0.17e-6)
    (3.0,0.115e-6)
    (3.1,0.095e-6)
    (3.2,0.075e-6)
    (3.3,0.055e-6)
    (3.4,0.055e-6)
    (3.5,0.055e-6)
    (3.6,0.055e-6)
    (3.7,0.055e-6)
    (3.8,0.040e-6)
    (3.9,0.020e-6)
    (4.0,0)
    (4.1,0)
};
\addlegendentry{5G code (1), $T(6_{384},0)$ (full syndrome)}

\addplot[
    very thick,
    mark=square,
    dashed,
    mark size=1.7pt,
    color=black,
] coordinates {
    (2.9,0.17e-6)
    (3.0,0.115e-6)
    (3.1,0.095e-6)
    (3.2,0.075e-6)
    (3.3,0.055e-6)
    (3.4,0.055e-6)
    (3.5,0.055e-6)
    (3.6,0.055e-6)
    (3.7,0.055e-6)
    (3.8,0.040e-6)
    (3.9,0.020e-6)
    (4.0,0)
    (4.1,0)
};
\addlegendentry{5G code (1), $T(4_{384},0)$}

\addplot[
    thick,
    mark=*,
    mark size=1.7pt,
    color=cyan!55!violet,
] coordinates {
    (2.9,3.42e-6)
    (3.0,2.58e-6)
    (3.1,2.13e-6)
    (3.2,1.25e-6)
    (3.3,0.87e-6)
    (3.4,0.68e-6)
    (3.5,0.50e-6)
    (3.6,0.35e-6)
    (3.7,0.35e-6)
    (3.8,0.26e-6)
    (3.9,0.14e-6)
    (4.0,0.09e-6)
    (4.1,0.055e-6)
};
\addlegendentry{5G code (1), $T(2_{384},2_{16})$}
\addplot[
      very thick,
    mark=*,
    mark size=1.7pt,
    color=red!75!black,
] coordinates {
    (3.2,0)
    (3.3,0)
    (3.4,0)
    (3.5,0)
    (3.6,0)
    (3.7,0)
    (3.8,0)
};
\addlegendentry{random (2), $T(8_{192},0)$ (full syndrome) }

\addplot[
    thick,
    dashed,
    mark=o,
    mark size=1.9pt,
    color=teal!75!green,
] coordinates {
    (3.2,0)
    (3.3,0)
    (3.4,0.03e-6)
    (3.5,0.01e-6)
    (3.6,0)
    (3.7,0)
    (3.8,0)
};
\addlegendentry{random (2), $T(3_{192},5_{16})$}

\end{axis}

\begin{axis}[
    at={(0,0)},
    anchor=south west,
    width=0.88\textwidth,
    height=3.85cm,
    xmin=2.85,
    xmax=4.15,
    ymin=-0.025e-6,
    ymax=0.22e-6,
    xtick={2.9,3.0,3.1,3.2,3.3,3.4,3.5,3.6,3.7,3.8,3.9,4.0,4.1},
    xticklabels={$2.9$,$3.0$,$3.1$,$3.2$,$3.3$,$3.4$,$3.5$,$3.6$,$3.7$,$3.8$,$3.9$,$4.0$,$4.1$},
    ytick={0,0.5e-7,1.0e-7,1.5e-7,2.00e-7},
    yticklabels={$0$,$0.5$,$1.0$,$1.5$,$2.0$},
    scaled y ticks=false,
    tick align=outside,
    tick style={black},
    xlabel={SNR dB},
    ylabel={zoom ($\times 10^{-7}$)},
    xlabel style={font=\small},
    ylabel style={font=\small},
    tick label style={font=\small},
    ymajorgrids,
    xmajorgrids,
    grid style={gray!25},
    axis line style={black},
]

\addplot[
    very thick,
    mark=*,
    mark size=1.7pt,
    color=red!75!black,
] coordinates {
    (3.2,0)
    (3.3,0)
    (3.4,0)
    (3.5,0)
    (3.6,0)
    (3.7,0)
    (3.8,0)
};

\addplot[
      thick,
    dashed,
    mark=o,
    mark size=1.9pt,
    color=teal!75!green,
] coordinates {
    (3.2,0)
    (3.3,0)
    (3.4,0.03e-6)
    (3.5,0.01e-6)
    (3.6,0)
    (3.7,0)
    (3.8,0)
};

\addplot[
    very thick,
    mark=square,
    dashed,
    mark size=1.7pt,
    color=black,
] coordinates {
    (2.9,0.17e-6)
    (3.0,0.115e-6)
    (3.1,0.095e-6)
    (3.2,0.075e-6)
    (3.3,0.050e-6)
    (3.4,0.050e-6)
    (3.5,0.050e-6)
    (3.6,0.050e-6)
    (3.7,0.050e-6)
    (3.8,0.040e-6)
    (3.9,0.020e-6)
    (4.0,0)
    (4.1,0)
};

\addplot[
      thick,
    mark=o,
    mark size=1.9pt,
    color=blue!70,
] coordinates {
    (2.9,0.17e-6)
    (3.0,0.115e-6)
    (3.1,0.095e-6)
    (3.2,0.075e-6)
    (3.3,0.050e-6)
    (3.4,0.050e-6)
    (3.5,0.050e-6)
    (3.6,0.050e-6)
    (3.7,0.050e-6)
    (3.8,0.040e-6)
    (3.9,0.020e-6)
    (4.0,0)
    (4.1,0)
};

\addplot[
    thick,
    mark=*,
    mark size=1.7pt,
    color=cyan!55!violet,
] coordinates {
    (2.9,3.42e-6)
    (3.0,2.58e-6)
    (3.1,2.13e-6)
    (3.2,1.25e-6)
    (3.3,0.87e-6)
    (3.4,0.68e-6)
    (3.5,0.50e-6)
    (3.6,0.35e-6)
    (3.7,0.35e-6)
    (3.8,0.26e-6)
    (3.9,0.14e-6)
    (4.0,0.09e-6)
    (4.1,0.055e-6)
};
\end{axis}
\end{tikzpicture}
\caption{UIBLER on a linear scale for 5G LDPC code and random code with different early termination criterias}
\label{FigET_UIBLER}
\end{figure}

 In 5G codes with all types of early termination: by all-layers (syndrome), $4$-layers  and reduced $4$-layers early terminations there {\bf were instances where decoder stopped with undetected errors such that the decoded information block is incorrect, despite that the syndrome (full, four layers or reduced four layers) is zero}, see Fig. \ref{FigET_UIBLER}. 
 We also see that 5G LDPC code with termination criteria $T(2_{384},2_{16})$ has relatively higher UIBLER, see Fig. \ref{FigET_UIBLER}. Further in CRC check block, see Fig.\ref{5g:scheme}  these frames are still detected as erroneous because they highly likely have incorrect CRCs. We exclude the further CRC check in our investigation because combined UIBLER for CRC and LDPC for such 24 bit CRC check is $2^{-24}10^{-7}$ is impossible to catch using ordinary computer.

We see that the reduced early termination $T(3_{192},5_{16})$  for the random code shows insignificant increase in undetected errors. In general, the tests of 5G code and random code shows that the reduction approach for syndrome calculation can
be used for LDPC codes in communication systems with retransmission and extra error-detection features like CRC.

\noindent{\it Acknowledgement.}
The authors would like to express their gratitude to Alexander Subach, who provided the first C++ implementation of reduced early termination condition  and confirmed its feasibility at the early stages of this study.

\section{Appendix A.}
The proof of Lemma \ref{Dub_lemma}.
\begin{proof}
Let $\pi(v, l)$ be a cyclic shift of a block $v$ by $l$ positions.

Firstly we note that for a block $u$ of size $q$, the shift by $q$ bits fixes the block any repeatition of the block $u$ and $Dub(u,\frac{\overline{q}}{q})$ in particular. 
Therefore is the following simplification of the actions of the cyclic shift  by $e+aq$, $e<q$ positions holds: \begin{equation}\label{ee1}\pi(Dub(u,\frac{\overline{q}}{q}),e+aq)=Dub(\pi(u,e),\frac{\overline{q}}{q}).\end{equation}

For matrix $E_{r\times n}$ and a block $$u=(u_1^1,u_1^2,\ldots,u_1^{q}, u_2^1,\ldots,u_2^{q},\ldots,u_n^1,\ldots,u_n^{q})$$ of length $nq$ consider the parity check equation for the word $Dub(u,\frac{\overline{q}}{q})$ in code $C(E,\overline{q})$:

\begin{equation}\label{ee2}\forall i\in\{1,\ldots,r\}
\sum _{j\in E_{ij}\neq -1}
\pi(Dub(u_{j}^1,\ldots,u_{j}^q),E_{ij} \mbox{ mod }\overline{q})=0,\end{equation}
where $0$ is of  length $\overline{q}$. 
From  equality (\ref{ee1}) we obtain: 
$$ \forall i\in\{1,\ldots,r\}
\sum _{j\in E_{ij}\neq -1}
\pi(Dub(u_{j}^1,\ldots,u_{j}^q),E_{ij} \mbox{ mod }\overline{q})=$$
$$
\sum _{j\in E_{ij}\neq -1}
Dub(\pi(u_{j}^1,\ldots,u_{j}^q),E_{ij} \mbox{ mod } q,\frac{\overline{q}}{q} ).$$

Since obviously the order of  the "repeating"\, operation $Dub$ is interchangeable with blockwise XOR we obtain:

$$\sum _{j\in E_{ij}\neq -1}
Dub(\pi(u_{j}^1,\ldots,u_{j}^q),E_{ij} \mbox{ mod } q,\frac{\overline{q}}{q} )=$$ $$Dub(\sum_{j\in E_{ij}\neq -1}\pi(u_{j}^1,\ldots,u_{j}^q),E_{ij} \mbox{ mod } q),\frac{\overline{q}}{q} ).$$

The repeating of a vector is  zero if and only if the vector is zero so 
\begin{equation}\label{ee3}\sum_{j\in E_{ij}\neq -1}\pi(u_{j}^1,\ldots,u_{j}^q),E_{ij} \mbox{ mod } q)=0\end{equation}

if and only if the parity check equalities  (\ref{ee2}) for $Dub(u)$ in $C(E,\overline{q})$ hold. 
4
We finish by noting that (\ref{ee3}) is the parity check equalities for a word $u$ in the code $C(E,q)$.
\end{proof}

The proof of Theorem \ref{T1}
\begin{proof}
1. Let $(u_1,\ldots,u_n)$ be a minimum weight codeword  of $C(E,q)$. By Lemma 1, we obtain that $Dub(u_1,\ldots,u_n,\overline{q}/q)$ is  a codeword of $C(E,\overline{q})$ of weight $\frac{\overline{q}}{q}d(C(E,q))$.

2. Let $v$ be a vector of $C(E,\overline{q})$:
$$(v_1^1,\ldots v_1^{\overline{q}/q}, v_2^1,\ldots,v_2^{\overline{q}/q},\ldots,v_n^1,\ldots,v_n^{\overline{q}/q}),$$
 where $v_i^j$, $i \in \{1,\ldots,n\}$, $j\in \{1,\ldots,\overline{q}/q\}$ are all of lengths $q$.
 Since $C(E,\overline{q})$ is a quasi-cyclic code, the following words are also its codewords:
$$(v_1^2,\ldots,v_1^1, v_2^2,\ldots,v_2^1\ldots,v_n^2,\ldots,v_n^1),$$
$$\ldots$$
$$(v_1^i,\ldots,v_1^{i-1}, v_2^i,\ldots,v_2^{i-1}\ldots,v_n^i,\ldots,v_n^{i-1}),$$
$$\ldots$$
$$(v_1^n,\ldots,v_1^{n-1}, v_2^n,\ldots,v_2^{n-1}\ldots,v_n^n,\ldots,v_n^{n-1}).$$

Summing all these vectors and $v$ modulo $2$, we obtain 
$$(\sum_{j=1}^n v_1^j,\ldots, \sum_{j=1}^n v_1^j, \sum_{j=1}^n v_2^j,\ldots \sum_{j=1}^n v_2^j,\ldots, \sum_{j=1}^n v_n^j,\ldots \sum_{j=1}^n v_n^j),$$
which is $Dub(( \sum_{j=1}^n v_1^j, \sum_{j=1}^n v_2^j,\ldots , \sum_{j=1}^n v_n^j,\overline{q}/q)$. This vector is in $C(E,\overline{q})$ by linearity and therefore $( \sum_{j=1}^n v_1^j, \sum_{j=1}^n v_2^j,\ldots , \sum_{j=1}^n v_n^j)=R(v,q,\overline{q})$ is in $C(E,q)$ by Lemma \ref{Dub_lemma}.

3. It is sufficient to prove for the case when $\overline{q}$ is $2q$. Consider a nonzero codeword $v$ of the code $C(E,\overline{q})$, which we represent as follows: 
$$(v_1^1,v_1^2, v_2^1,v_2^2,\ldots,v_n^1,v_n^2),$$
where all vectors $v_i^j$, $i \in \{1,\ldots,n\}$, $j\in \{1,2\}$ are of lengths $q$.

We have several cases there.

Case A. Let $v$ be $Dub(v_1^1,\ldots,v^1_n,2)$. Then by Lemma 1, $(v_1^1,\ldots,v^1_n)$ is a nonzero codeword of $C(E,q)$ and its weight is twice less than the weight of $v$.

Therefore, taking  $v$ to be a minimum weight codeword of $C(E,\overline{q})$, we have $$d(C(E,q)) \leq d(C(E,\overline{q})) /2\leq d(C(E,\overline{q})).$$

Case B. Let $v$ be such that $v_i^0\neq v_i^1$ for some $i$. Then $R(v,q,\overline{q})$ is a nonzero vector of $C(E,q)$ with weight not greater than the weight of $v$.  If $v$ is a minimum weight codeword of $C(E,\overline{q})$, we have that
$$d(C(E,q))\leq d(C(E,\overline{q})).$$ \end{proof}

\end{document}